\documentclass{article}%
\usepackage{amsfonts}
\usepackage{amsmath}
\usepackage{amssymb}
\usepackage{graphicx}%
\providecommand{\U}[1]{\protect\rule{.1in}{.1in}}
\newtheorem{theorem}{Theorem}
\newtheorem{acknowledgement}[theorem]{Acknowledgement}

\newtheorem{corollary}[theorem]{Corollary}

\newtheorem{definition}[theorem]{Definition}
\newtheorem{example}[theorem]{Example}

\newtheorem{lemma}[theorem]{Lemma}

\newtheorem{proposition}[theorem]{Proposition}
\newtheorem{remark}[theorem]{Remark}

\newenvironment{proof}[1][Proof]{\noindent\textbf{#1.} }{\ \rule{0.5em}{0.5em}}
\begin{document}

\title{ Numerical stability of generalized entropies}
\author{Gy\"{o}rgy Steinbrecher\\University of Craiova, A. I. Cuza 13, 200585 Craiova, Romania.\\Email: gyorgy.steinbrecher@gmail.com
\and Giorgio Sonnino\\Department of Theoretical Physics and Mathematics, Universit\'{e}\\Libre de Bruxelles (ULB), Campus Plain CP 231, Boulevard de \\Triomphe, 1050 Brussels, Belgium \& \\Royal Military School (RMS)\\Av. de la Renaissance 30, 1000 Brussels - Belgium.\\E-mail: gsonnino@ulb.ac.be}
\maketitle

\begin{abstract}
In many applications, the probability density function is subject to
experimental errors. In this work the continuos dependence of a class of
generalized entropies on the experimental errors is studied. This class
includes the C. Shannon, C. Tsallis, A. R\'{e}nyi and generalized R\'{e}nyi
entropies. By using the connection between R\'{e}nyi or Tsallis entropies, and
the \textit{distance} in a the Lebesgue functional spaces, we introduce a
further extensive generalizations of the R\'{e}nyi entropy. In this work we
suppose that the experimental error is measured by some generalized $L^{p}$
distance. In line with the methodology normally used for treating the so
called \textit{ill-posed problems}, auxiliary stabilizing conditions are
determined, such that small - in the sense of $L^{p}$ metric - experimental
errors provoke small variations of the classical and generalized entropies.
These stabilizing conditions are formulated in terms of $L^{p}$ metric in a
class of generalized $L^{p}$ spaces of functions. Shannon's entropy requires,
however, more restrictive stabilizing conditions.


\end{abstract}

\section{\bigskip Introduction}

The generalizations of the Boltzmann-Shannon entropy (BSE) \cite{shannon}, the
R\'{e}nyi entropy (RE) \cite{Renyi1}-\cite{Renyi4} and the Tsallis entropy
(TE) \cite{Tsallis1}-\cite{TsallisGelMann} were introduced, independently, by
an axiomatic respectively physical approach. These entropies are algebraically
related to the same norm, or \textit{distance functional}, in suitable $L^{p}$
spaces \cite{SonninoSteinbrGRE}. There are many applications of this class of
generalized entropies \cite{TsallisBook}, \cite{TsallisGelMann},
\cite{Klimontovich}-\cite{RenyiDivergence}. The problem of mathematical
naturalness of the generalized entropies, in the sense of category theory, is
treated in \cite{SGySASGCategory}. Since in the generic cases, the classical
and generalized entropies, are functionals defined in infinite-dimensional
$L^{p}$ spaces, where $p$ is exactly the parameter in the RE and TE (see
\cite{SonninoSteinbrGRE} and references therein), the first problem is related
to the correct domain of definition, the finiteness of the defining
functionals. Here, we shall use the general, Lebesgue $L^{p}$ space formalism
\cite{SonninoSteinbrGRE}, \cite{SGySASGCategory}.

In applications the probability density function (PDF) is known only
approximately, so the second, practical, problem concerns the numerical
stability of the generalized and the classical entropies when the PDF is known
only approximately. This stability problem is the main subject of the our
study. This problem was studied also in relation to the convergence in Central
Limit Theorem in \cite{CLTapproximation}. Since the PDF is always integrable,
one of the natural measure of the approximation of PDF is the $L^{1}$
\textit{distance} between the exact and approximate PDFs. This
\textit{distance} is always finite. In the case of discrete probability
distribution, the numerical stability is studied in \cite{Lesche}, where
numerical stability is formulated in terms of $l^{1}$ distance for the
input-data, and of \emph{relative error of the entropy} for the output-data.
We are pushed to extend the concepts of \textit{finiteness} and
\textit{numerical stability} by using the general formalism of the probability
theory in general measure spaces \cite{Rudin}-\cite{EStein}, that include the
$l^{1}$ space. Generally, when the only informations is related to the $L^{1}$
norm, we have no control on the other $L^{p}$ distances, both for RE and TE.
Both TE and RE (in the sequel, referred to as RTE) may be expressed in terms
of $L^{p}$-distances of the PDFs. However in general, despite the $L^{1}%
$-convergence of the approximate PDF sequence, different $L^{p}$-convergences
(and topologies) are not equivalent \cite{Rudin}-\cite{EStein}, in
contradistinction to finite dimensional case where all of the $l^{p}$ norms
are equivalent. It is clear that in the generic infinite-dimensional case, the
continuity of entropies in the case of approximation of PDFs cannot be deduced
without evoking extra stabilizing conditions, conditions that make the problem
similar to finite dimensional case: in the subset defined by stabilization
conditions, the convergence say in $L^{1}$ norm implies convergence in a
family of $L^{p}$ norm (see below). In our work we follow the methods adopted
for treating the so called \textit{ill-posed problem} (in the sense of J.
Hadamard). This method is used in the field of strong interactions in physics
not only for establishing numerical methods \cite{CiulliNumerical1},
\cite{CiulliNumerical2}, but also for formulating the problem in more rigorous
terms \cite{CiulliStabProblems}-\cite{NenciuHamonicMeasure}. Likely, this
method can be reformulated appropriately to tackle and solve our problem. In
this work we shall firstly find auxiliary \textit{stabilizing conditions}
e.g., by imposing that RTE is finite for a particular value of the parameter
$p$. Secondly, we shall prove that, by imposing these stabilizing conditions,
for a large range of the parameter $p$, the \textit{absolute error} for the
RTE is controlled by the $L^{1}$ error in the input-PDFs. The reason to use
the absolute error instead of the relative error is discussed in
\cite{JizbaArimitsu}. This choice is also justified in a counter example shown
and discussed in the Subsection \ref{markerSubsectCounterxDiscrete}. Indeed,
it is possible to have very small relative errors in TE even when TE diverges.
By using the absolute error as a measure of numerical stability, it follows
that the quantities that are related by continuos functions, like RTE, as well
as the corresponding $L^{p}$ distances, must have similar numerical stability
properties. Note that special attention should be addressed to the limit case
of BSE: the conditions that stabilize RTE are not sufficient to stabilize BSE,
more stabilizing conditions are necessary and in the stability proof the
powerful functional analytic methods used previously in high-energy physics
are involved \cite{CiulliStabProblems}-\cite{NenciuHamonicMeasure}.

The classical definitions of the quantity of information, or of the degree of
randomness, at least in the case of discrete (atomic) phase space, tacitly
assume the invariance of the quantity of information, or entropy, under
measure preserving groups that in classical simple cases are the permutations
\cite{shannon}-\cite{RenyiWT}, \cite{Tsallis1} -\cite{TsallisGelMann}.
\emph{This symmetry assumption applies to the case of more abstract
definitions i.e. the previous definitions of the entropy are invariant under
the group of measure preserving transformations. This reflects our complete
lack of knowledge on the possible dynamical behavior \cite{maszczyk}. Larger
symmetry group means less information.} Nevertheless, for real physical
systems, for instance when the phase space is a direct product of subspaces
having completely different physical interpretation and different mathematical
structures (e.g., the case when the total system splits into \textit{driving}
and \textit{driven} system \cite{SonninoSteinbrGRE}), we have already an
information. In such a situation, the general measure preserving
transformations has no physical interpretation and should not be considered as
a \textit{true}, we mean \textit{physical}, symmetry \cite{SGySASGCategory}.
Consequently, we need a new generalization of the entropy able to take into
account this direct product structure. The first step in this direction can be
found in a previous work cited in Ref. \cite{SonninoSteinbrGRE} where a new
class of extensive generalization of R\'{e}nyi's entropy is introduced to
analyze the case of a phase space, which may be split as a direct product of
two sub-spaces (note that the approach illustrated in this work applies
equally to the case where the phase-space may be split in an arbitrary, but
finite, number of direct products of sub-spaces, like in data with tensor structures).

In this work our aim is to study the stability properties of this class of
entropies. As in the previous work \cite{SonninoSteinbrGRE}, our guiding line
is to re-interpret the generalized entropy in terms of distance in some
generalization of classical \cite{Rudin}-\cite{LuschgiPages} or anisotropic
$L^{p}$ spaces \cite{Besov}.

The main results is the possibility to stabilize RTE or GRE for a large
domains of the defining parameters by imposing the finiteness of RTE or GRE
for some special values of parameter. An interesting result is that for the
stabilization of BSE we need more stabilizing conditions, the BSE is more
sensitive compared to general RTE.

The properties of $L^{p}$ spaces, and associated "distances", will be used for
setting the numerical stability properties when auxiliary \textit{stabilizing
conditions} are imposed. In this work the stabilizing conditions are expressed
as \textit{finiteness} of some classical or generalized (anisotropic
\cite{Besov}) Lebesgue space distances. In some particular cases, these
stabilizing conditions may be found by studying directly the partial
differential equations for the PDFs. This allows obtaining bounds on the
Sobolev space norm and, by using the embedding theorems, getting information
on the $L^{p}$ norms \cite{Besov}, \cite{ESteinSingularIntegrals}. Another
possibility is the study of the \textit{heavy tail index \cite{RLASATO},}
\cite{sgw}, \cite{sgbw}.

The main results is the possibility to stabilize RTE or GRE for a large
domains of the defining parameters by imposing the finiteness of RTE or GRE
for some special values of parameter. An interesting result is that for the
stabilization of BSE we need more stabilizing conditions, the BSE is more
sensitive compared to general RTE.

The work is organized as follows. In Section
\ref{markerSectionStabltyClassicalEntropies},
Subsection~\ref{markersubsectDefRenyiTsallisPseudonorm}, we relate the RTE to
a norm or pseudo-norm in a suitable defined $L^{p}$ spaces. In the
Subsection~\ref{markerSubsectDomainConvClassicalNentr}, the domain of
definition and the problem of continuity of the Shannon, R\'{e}nyi and Tsallis
entropies are studied, when the PDF is approximated in the sense of distance
in $L^{p}$ spaces, with particular emphasis on the natural (from probabilistic
point of view) $L^{1}$ distance. In particular, we discuss about different
concepts of stability (\textit{continuity}, in mathematical terms) of the
entropies when the PDF is known only approximately. We find
\textit{sufficient} stabilizing conditions for generic RTE. It is also proven
(by counter examples) that in the case of BSE, when a sequence of PDF
$\rho_{n}$ converges in the $L^{1}$ norm to the PDF$\rho$, extra stabilizing
conditions are necessary to ensure the continuity of BSE. We determine the
stability of BSE under suitable stabilizing conditions. We emphasize the
central role of the functional $Z_{q}[\rho]$ (see below,
\cite{SGySASGCategory}) that appears in the both definition of RTE and BSE.
The special analytic properties of this functional as well as its relation to
the distances in $L^{p}$ spaces are used in the proofs.

In Section~\ref{markerSectionGRE}, we introduce a further, natural, extension
of the Generalized R\'{e}nyi Entropy (GRE), studied in the previous work
ref.(\cite{SonninoSteinbrGRE}, \cite{SGYSGstructure}, \cite{SGySASGCategory}),
its general properties are studied and some physical applications are
presented. In the Section \ref{markerSectionGREboundStabiltiy} the finiteness
and continuity results from Section
\ref{markerSectionStabltyClassicalEntropies} are extended to the most general
case of GRE assuming suitable stabilizing conditions.

We emphasize that the stabilizing conditions are expressed in the terms of
measure theoretic concepts. Hence, the results can equally be applied to both
continuos and discrete probability distributions, on a general measured space.
No continuity or differentiable properties of the PDF are assumed.

Mathematical details can be found in the Appendix~\ref{markerSectionAppndx}.
In Subsection\ref{markerSubsectExtrapolTheorem}, the powerful analytic
extrapolation method is exposed. This method is necessary to prove the
stability of BSE. In Subsection~\ref{markerSubsectCalculSumeIntegrale} and In
Subsection \ref{markerSubsectionPropertiesOfFunctionalDpmf} we found details
on the proofs of counter examples and we expose the properties of a class of
metric vector spaces, related to the definition and properties of GRE,
respectively. In Subsection~\ref{markerAppndxLogConvexity}, the logarithmic
convexity of a class of functionals related to RTE and GRE is presented. These
functionals are used in the main proofs of the bound and the numerical
stability of the RTE and GRE. For easy reference, since our proof will be
generalized to the study of the GRE, some of the known logarithmic convexity
properties of the classical $L^{p}$\ norms \cite{EStein} are demonstrated. The
long proofs of stability of the classical Shannon-Boltzmann entropy and of GRE
are also presented in detail.

\section{Finiteness and continuity aspects of the Shannon, R\'{e}nyi and
Tsallis entropies\label{markerSectionStabltyClassicalEntropies}}

\subsection{Definitions of generalized entropies in term of norms and pseudo
norms \label{markersubsectDefRenyiTsallisPseudonorm}}

Let's consider a standard measure space $(\Omega,\mathcal{A},m)$ where
$\Omega$ is the phase space of a natural system (usually $\mathbf{R}^{n}$, or
in simplest probability models, some finite or denumerable set, or the space
of all Brownian trajectories), $\mathcal{A}$ is some $\sigma-$algebra
generated by a family-subsets of $\Omega$ (that are usually viewed as the
collection of the observable events in a probability space), and \thinspace$m$
is a $\sigma$-finite measure defined on $\mathcal{A}$. In this framework the
R\'{e}nyi divergence is expressed by R\'{e}nyi entropy \cite{SGySASGCategory}.

In the simplest cases of discrete (atomic) probability spaces, the measure
(not necessary finite) $m$ is a counting measure. Note that in this simple
case the counting measure is invariant under the group of permutations on
$\Omega$. In general, the measure $m$ is chosen to be invariant under some
symmetry group of the physical system under study, but there are simple
examples for measures with trivial measure preserving group
\cite{SGySASGCategory}. This formalism allows to consider mixture the discrete
and continuos distributions in many variables, and to express the R\'{e}nyi
divergence by RE \cite{SGySASGCategory}. In the sequel, we consider only
probability measures defined on $(\Omega,\mathcal{A})$ that are continuos with
respect to some fixed measure $m$, that means that for $A\in\mathcal{A}$ the
probability $p(A)$ can be represented by probability density function (PDF)
$\rho$ as follows
\begin{align}
p(A) &  =\int\limits_{A}\rho(x)dm(x);~A\subset\Omega\label{LL1}\\
\rho(x) &  \geq0;~\int\limits_{\Omega}\rho(x)dm(x)=1\label{LL1.1}%
\end{align}
or in shorthand notation $dp(x)=\rho(x)dm(x)$. In this framework, the
classical BSE is given by
\begin{equation}
S_{cl}[\rho]=-\int\limits_{\Omega}\log\left[  \rho(x)\right]  \rho
(x)dm(x)\label{LL2}%
\end{equation}
The next subsequent generalizations, R\'{e}nyi's \cite{Renyi1}-\cite{Renyi4}
and Tsallis' \cite{Tsallis1}-\cite{TsallisGelMann}, \cite{Tsallis2} entropies
can be reformulated in the terms of distances \cite{SonninoSteinbrGRE} in the
$L^{p}(\Omega,dm)$ spaces of the distribution functions as follows. They
correspond to the norms $\left\Vert \rho\right\Vert _{p}$ \cite{ReedSimon} for
$p\geq1$ and the pseudo-norm $N_{p}[\rho]$ \cite{Rudin}, \cite{LuschgiPages},
\cite{SonninoSteinbrGRE} for $0<p<1$, respectively:
\begin{align}
\left\Vert \rho\right\Vert _{p,m} &  =\left[  {\int\limits_{\Omega}}\left[
\rho(x)\right]  ^{p}dm(x)\right]  ^{\frac{1}{p}};~p\geq1\label{LL3}\\
N_{p,m}[\rho] &  ={\int\limits_{\Omega}}\left[  \rho(x)\right]  ^{p}%
dm(x);~0<p\leq1\label{LL4}%
\end{align}
Details can be found in Ref.~\cite{SonninoSteinbrGRE}, Section III. In order
to have a unified treatment for both cases $p\lessgtr1$, we introduce the
following basic \textit{condensed} notations:
\begin{align}
i(p) &  :=1/p~~for~p\geq1;~i(p):=1~for~p\leq1\label{LL4.1}\\
D_{p,m}[f] &  :=\left[  {\int\limits_{\Omega}}\left\vert f(x)\right\vert
^{p}dm(x)\right]  ^{i(p)}\label{LL4.2}\\
Z_{w,m}[f] &  :={\int\limits_{\Omega}}\left\vert f(x)\right\vert
^{w}dm(x);~\operatorname{Re}(w)>0\label{LL4.3}%
\end{align}

\begin{remark}
\label{MarkerRemarkN_p_Definition} The function $D_{p,m}[f]$ is a generalized
distance from the "point" $f$ (in the function space) to the origin. Note that
particular cases of $D_{p,m}[f]$ are $\left\Vert f\right\Vert _{p,m}$
($p\geq1)$ and $N_{p,m}[\rho]$ (for $0<p\leq1$) defined in Eqs.(\ref{LL3},
\ref{LL4}). The functional $f\rightarrow$ $D_{p,m}[f]$ has following basic
metric properties:
\begin{align}
D_{p,m}[f_{1}+f_{2}]  &  \leq D_{p,m}[f_{1}]+D_{p,m}[f_{2}]\label{LL4.4}\\
\left\vert D_{p,m}[f_{1}]-D_{p,m}[f_{2}]\right\vert  &  \leq D_{p,m}%
[f_{1}-f_{2}]\label{LL4.41}\\
D_{p,m}[\alpha f]  &  =\left\vert \alpha\right\vert ^{s(p)}\ D_{p,m}%
[f]~;~\alpha\in\mathbb{R}\label{LL4.5}\\
s(p)  &  =pi(p) \label{LL4.6}%
\end{align}
For $p\geq1$ the functional $D_{p,m}[f]=\left\Vert f\right\Vert _{p,m}$ (i.e.,
the classical $L^{p}$ norm), and for $0<p\leq1$ the functional $D_{p,m}%
[f]=N_{p,m}[f]$ (i.e., the "exotic" $L^{p}$ pseudo norm from
Refs.(\cite{Rudin}), \cite{LuschgiPages}). The "precision of approximation
$\rho_{n}$ of the true PDF $\rho$ " is quantified by $D_{p.m}[\rho-\rho_{n}]$
for some $p>0$. A possible natural choice is $p=1$, because $D_{1,m}[\rho
-\rho_{n}]=\int_{\Omega}\left\vert \rho-\rho_{n}\right\vert dm$ is at least
well defined due to Eq.(\ref{LL1.1}).
\end{remark}

The corresponding generalized entropies proposed by R\'{e}nyi \cite{Renyi1}%
-\cite{Renyi4} $S_{R,q,m}$ and by C. Tsallis \cite{Tsallis1}%
-\cite{TsallisGelMann} $S_{T,q,m}$, are given respectively by
\begin{align}
S_{R,q,m}[\rho]  &  =\frac{1}{1-q}\log Z_{q,m}[\rho]\label{LL5}\\
S_{R,q,m}[\rho]  &  =\frac{1}{(1-q)i(q)}\log D_{q,m}[\rho]\label{LL6}\\
S_{T,q,m}[\rho]  &  =\frac{1}{1-q}\left[  1-Z_{q,m}[\rho]\right]  =\frac
{1}{1-q}\left\{  1-\left[  D_{q,m}[\rho]\right]  ^{1/i(q)}\right\}
\label{LL7}%
\end{align}
From Eqs.(\ref{LL2}, \ref{LL4.3}) results%
\begin{equation}
S_{cl}[\rho]=-\frac{d}{dw}Z_{w,m}(\rho)|_{w=1} \label{LL71.}%
\end{equation}
when it exists. The functionals $S_{T,q,m}[\rho]$, $S_{T,q,m}[\rho]$,
$D_{q,m}[\rho]$ and $Z_{q,m}[\rho]$ are algebraically related so they contain
exactly the same amount of information. Functional $Z_{q,m}[\rho]$ is the best
candidate for including the entropy related functionals in a powerful
formalism of category theory \cite{SGySASGCategory}. This functional possesses
also the remarkable analytic and log-convex properties (see below). In the
following proof, we shall use successively the representation of the entropies
by functionals $Z_{q,m}[\rho]$ and the distance function $D_{q,m}[\rho]$. From
the previous definitions we get the following result:

\begin{remark}
\label{markerREMconvergenceEquivallence}Let $\ \rho_{n},\rho\in L^{p}%
(\Omega,dm)\cap L^{1}(\Omega,dm)$, $p\neq1$ with ${\int\limits_{\Omega}}%
\rho_{n}(x)dm(x)={\int\limits_{\Omega}}\rho(x)dm(x)=1$. \ The following
convergencies are equivalent:
\begin{align*}
&  D_{q,m}[\rho_{n}]\underset{n\rightarrow\infty}{\rightarrow}D_{q,m}[\rho]\\
&  Z_{q,m}[\rho_{n}]\underset{n\rightarrow\infty}{\rightarrow}Z_{q,m}[\rho]\\
&  S_{T,q,m}[\rho_{n}]\underset{n\rightarrow\infty}{\rightarrow}S_{T,q,m}%
[\rho]
\end{align*}
If in addition $D_{q,m}[\rho]>0$ then the previous convergences are equivalent
to the statement $S_{R,q,m}[\rho_{n}]\underset{n\rightarrow\infty}%
{\rightarrow}S_{R,q,m}[\rho]$.
\end{remark}

\begin{remark}
\label{markREM_counting_measure} Note that when $\Omega$ is a finite or
denumerable set, if we denote by $p_{k}$ the probabilities of element
$x_{k}\in$ $\Omega$, with the measure $m$ denoting the counting measure (equal
to number of elements in a subset) on the space $\Omega$, then from the
previous Eqs.(\ref{LL3}-\ref{LL6}) we get the original definitions
Refs.~\cite{Renyi1}-\cite{Renyi4}, \cite{Tsallis1}-\cite{TsallisGelMann}
\begin{align}
S_{R,q}[\rho]  &  =\frac{1}{1-q}\log\sum\limits_{k}p_{k}^{q}\label{LL8.1}\\
S_{T,q}[\rho]  &  =\frac{1}{1-q}\left[  \sum\limits_{k}p_{k}^{q}-1\right]
\end{align}

\end{remark}

\subsection{Domain of definition (finiteness) and continuity of the Tsallis,
R\'{e}nyi and Shannon entropies. \label{markerSubsectDomainConvClassicalNentr}%
}

\subsubsection{Generic case: Tsallis and R\'{e}nyi entropies}

\paragraph{Domain of finiteness of RTE's.}

In the study of the stability (or continuity) properties we consider the
general case, when the measure $m(\Omega)$ is not necessary finite and its
support does not reduce, necessarily, to a denumerable set. By simple
inspection of the Eqs.(\ref{LL4.3}, \ref{LL5}-\ref{LL7}) we conclude that the
R\'{e}nyi and Tsallis entropy (RTE) with same $q$ as well as the functionals
$D_{q,m}$ and $Z_{q,m}$, are finite or infinite, simultaneously. The problem
related to the domain of definition and continuity of the generalized
entropies, when some approximations of the probability density are used, is
reduced to the problem of finiteness of continuity of the norm Eq.(\ref{LL3}),
pseudo norms Eq.(\ref{LL4}) or the functional $Z_{w}(\rho)$ from
Eq.(\ref{LL4.3}), respectively. Note that in the case of general measure
space, when the measure $m$ is neither atomic nor probabilistic, for instance
when $\Omega=\mathbb{R}$ and $dm(x)=dx$, when $q_{1}\neq q_{2}$ the norms
$\left\Vert \rho\right\Vert _{q_{1}}$, $\left\Vert \rho\right\Vert _{q_{2}}$
as well as the pseudo norms $N_{q_{1}}[\rho]$, $N_{q_{2}}[\rho]$, are not
equivalent (the finiteness or convergence of a sequence $\rho_{n}$ in the norm
$\left\Vert \rho_{n}\right\Vert _{q_{1}}$ is not related to finiteness or
convergence in the norm $\left\Vert \rho_{n}\right\Vert _{q_{2}}$
\cite{ReedSimon}, \cite{Rudin} and there is no general inequality relating
them). Consider the following example.

\begin{example}
Let $\Omega=\mathbb{R}$, $dm(x)=dx$ and consider the PDF $\rho_{1}(x)$, such
that $\int\limits_{\mathbb{R}}\rho_{1}(x)dx=\left\Vert \rho_{1}\right\Vert
_{1}=1$ and even more $\int\limits_{\mathbb{R}}\left\vert \rho_{1}(x)\log
\rho_{1}(x)\right\vert dx<\infty$, \ $\int\limits_{\mathbb{R}}|\rho
_{1}(x)|^{r}dx<\infty$ for all $r>0$. Denote $\rho_{\lambda}(x):=\lambda
\rho_{1}(\lambda x)$. \ Despite $\left\Vert \rho_{\lambda}\right\Vert _{1}%
=1$\ observe that \ when $\lambda\rightarrow\infty$ \ for $p>1$ the RTE
diverges, for $0<p<1$ the RE diverges. In the limit $\lambda\rightarrow0$
\ for $p<1$ the RTE\ diverges and for $p>1$ the RE\ diverges. \ The BSE of
$\rho_{\lambda}(x)$ diverges both in \ limits $\lambda\rightarrow\infty$ and
$\lambda\rightarrow0$.
\end{example}

It follows that at first sight the answer to the previous questions is
tautological. But if we add the physical conditions Eqs~(\ref{LL1.1}) we can
obtain more information on the finiteness and convergence. The following lemma
comes from the well known results \cite{EStein}, when $p>1$. For easy
reference we give an elementary proof, which will be generalized to the study
of GRE (Generalized R\'{e}nyi entropies).

\begin{lemma}
\label{markerInterpolationLemma}Consider the measure space $(\Omega
,\mathcal{A},m)$ and let $f\in L^{1}(\Omega,dm)\cap L^{s}(\Omega,dm)$,
$s\neq1$, $s>0$\ with the property \bigskip%
\begin{align}
{\int\limits_{\Omega}}\left\vert f(x)\right\vert dm(x)  &  \leq a_{1}%
\label{LL8.1b}\\
{\int\limits_{\Omega}}\left\vert f(x)\right\vert ^{s}dm(x)  &  \leq a_{s}
\label{LL8.1c}%
\end{align}
Then for all $r$ in the range
\[
\min~(1,s)<r<\max\ (1,s)
\]
we have also the bound
\begin{equation}
{\int\limits_{\Omega}}\left\vert f(x)\right\vert ^{r}dm(x)\leq a_{1}%
^{\frac{s-r}{s-1}}a_{s}^{\frac{r-1}{s-1}} \label{LL8.1d}%
\end{equation}

\end{lemma}

\noindent This Lemma is a particular case of the Theorem
\ref{markerTheoremRieszThorinGen}, from the Appendix
\ref{markerAppndxLogConvexity}. \noindent Before studying the continuity (also
called stability) of the generalized entropies with respect to small
variations of the input PDF, firstly we should discuss the problem of the
finiteness of the generalized entropies. The singular case of classical
entropy will be discussed later. From the previous Lemma
\ref{markerInterpolationLemma} results the following

\begin{proposition}
\label{markerPropositionFinitness} Let $p\neq1$ and $\min~(1,p)<r<\max
\ (1,p)$. For the PDF $\rho$ obeying Eq.(\ref{LL1.1}), $\rho\in L^{1}%
(\Omega,dm)\cap L^{p}(\Omega,dm)$, and ${\int\limits_{\Omega}[}\rho
(x)]^{p}dm(x)\leq a_{p}$ we have
\begin{equation}
{\int\limits_{\Omega}[}\rho(x)]^{r}dm(x)\leq a_{p}^{\frac{r-1}{p-1}}
\label{LL8.1f}%
\end{equation}
Consequently if the norm $\left\Vert \rho\right\Vert _{p}$ is finite for some
fixed index index $p>1$, then $\left\Vert \rho\right\Vert _{r}$, $S_{R,r}%
[\rho]$, $S_{T,r}[\rho]$ remain finite for all $r$ in the range $1<r\leq p$.
If for some fixed $p$, with $0<p<1$ the pseudonorm $N_{p}[\rho]$ is finite,
then for all $r$ in the range $p\leq r<1$ also $N_{r}[\rho]$, $S_{R,r}[\rho]$,
$S_{T,r}[\rho]$ remain finite.

\begin{proof}
By setting in Lemma \ref{markerInterpolationLemma} $f(x)=\rho(x)$ from the
normalization condition Eq.(\ref{LL1.1}) results $a_{1}=1$. Eq.(\ref{LL8.1f})
results directly from Eq.(\ref{LL8.1d}). The finiteness of entropies results
from Eqs.(\ref{LL4.3}, \ref{LL5}, \ref{LL7}).
\end{proof}
\end{proposition}

When $\Omega=\mathbb{R}$ and $dm(x):=dx$, the pseudo-norm $N_{p}[\rho]$, and
the generalized entropies $S_{R,p}[\rho]$, $S_{T,p}[\rho]$ are divergent for
low values of $p<p_{0}<1$ (or, in practical estimations, they have large
fluctuations), we may argue that the PDF has a heavy tail: $\rho
(x)\underset{|x|\rightarrow\infty}{\asymp}|x|^{-1/p_{0}}$, like in models of
stochasticity-induced instability \cite{RLASATO}, \cite{sgbw}, \cite{sgw}. If
there exists some $p_{1}$ such that for $p>p_{1}>1$ the entropies, and both
norm and $S_{R,p}[\rho]$, $S_{T,p}[\rho]$ are divergent, this suggests that
the PDF has an integrable singularity of the type $|x-x_{0}|^{-1/p_{1}}$, for
instance in the models of noise driven intermittency \cite{intermittency}.
Similarly, in the case of PDF defined in higher dimensional, anisotropic space
the maximal domain of definition of generalized R\'{e}nyi entropies
\cite{SonninoSteinbrGRE} is related to more complicated singularity structure
and asymptotic behavior of the multivariate PDF.

In the continuation let us suppose to have an exact PDF $\rho\in L^{p}%
(\Omega,dm)\cap L^{1}(\Omega,dm)$, with $p\neq1$, which is approximated by an
approximant sequence $\rho_{n}\in L^{p}(\Omega,dm)\cap L^{1}(\Omega,dm)$, and
$\rho$, $\rho_{n}$ satisfy the Eq.(\ref{LL1.1}). We also suppose that the
"true" limit PDF $\rho$ exists, so $\rho_{n}$ is a Cauchy sequence, in $L^{1}$.

\begin{remark}
Our approach on the stability problem is different from Ref.\cite{Lesche},
where no stabilizing conditions are imposed. There are simple counter examples
\ref{markerSubsectCounterxDiscrete} of sequences of probability distributions
on $\mathbb{N}$ that are convergent in $\ l^{1}$, consequently bounded and
convergent in $l^{p}$ ($p>1)$ norm but the BSE diverges, despite it is Lesche
stable (see below). We mention that in the case when we restrict ourselves to
discrete distributions, from all of the convergence results based on
stabilizing conditions it follows also the Lesche stability. The counter
example from \ref{markerSubsectCounterxContinous} proves that there exist
sequences of PDF's $\rho_{n}$ that are convergent in all the spaces
$L^{p}([0,1/2],dx)$ with $0<p\leq1$, nevertheless the BSE diverges.
\end{remark}

\paragraph{Continuity of RTE's when the PDF is approximated in $L^{1}$ norm.}

Suppose that for some $p\neq1$ we have $L^{p}$ bounds%
\begin{align*}
{\int\limits_{\Omega}}\left\vert \rho_{n}(x)\right\vert ^{p}dm(x)  &  \leq
b_{1}\\
{\int\limits_{\Omega}}\left\vert \rho(x)\right\vert ^{p}dm(x)  &  \leq b_{2}%
\end{align*}

For instance, such kind of bounds could be obtained by the technique used in
the study of the heavy tail phenomena in random affine processes \cite{sgw},
\cite{sgbw}. The quality of the approximation is quantified in the $L^{1}$
norm $\left\Vert \rho_{n}-\rho\right\Vert _{L_{1}}$. By using Eq.(\ref{LL4.4})
it follows
\begin{align}
\int\limits_{\Omega}\left\vert \rho_{n}(x)-\rho(x)\right\vert ^{p}dm(x)  &
\leq b\label{LL8.37.1}\\
{\int\limits_{\Omega}}\left\vert \rho_{n}(x)-\rho(x)\right\vert dm(x)  &
\leq\varepsilon_{n}\underset{n\rightarrow\infty}{\rightarrow}0
\label{LL8.37.2}%
\end{align}
Here $b=b_{1}+b_{2}$ for $0<p<1$ and $b=\left(  b_{1}^{1/p}+b_{2}%
^{1/p}\right)  ^{p}$, see Eq.(\ref{LL4.2}). By using Eqs.(\ref{LL8.37.1},
\ref{LL8.37.2}) and the Lemma \ref{markerInterpolationLemma}, with
$a_{1}=\varepsilon_{n}$, $a_{p}=b$, with $f(x):=\rho_{n}(x)-\rho(x\mathbf{)}$,
we get for $\min~(1,p)<r<\max\ (1,p)$%

\begin{align}
{\int\limits_{\Omega}}\left\vert \rho_{n}(x)-\rho(x)\right\vert ^{r}dm(x)  &
\leq\varepsilon_{n}^{\frac{p-r}{p-1}}b^{\frac{r-1}{p-1}}\label{LL8.37.3}\\
\left[  D_{r,m}[\rho_{n}-\rho]\right]  ^{\frac{1}{i(p)}}  &  \leq
\varepsilon_{n}^{\frac{p-r}{p-1}}b^{\frac{r-1}{p-1}} \label{LL8.37.3.0}%
\end{align}
\noindent From Eqs.(\ref{LL8.37.3.0}, \ref{LL4.2}, \ref{LL4.3}, \ref{LL4.41})
we obtain
\begin{align}
\left\vert D_{r,m}[\rho_{n}]-D_{r,m}[\rho]\right\vert  &  \leq\delta
_{n}\underset{n\rightarrow\infty}{\rightarrow}0\label{LL8.37.3.1}\\
\delta_{n}  &  =\left[  \varepsilon_{n}^{\frac{p-r}{p-1}}m^{\frac{r-1}{p-1}
}\right]  ^{i(p)} \label{LL8.37.3.2}%
\end{align}
\noindent Taking into account remark \ref{markerREMconvergenceEquivallence} we
can summarize the previous results Eq.(\ref{LL8.37.3.1}) as follows

\begin{proposition}
\label{markerPropstionRenyiContinuity} Suppose that, for some $p\neq1$, we
have boundednes of $\left\vert \rho_{n}(x)-\rho(x)\right\vert $ in the $L^{p}$
norm (Eq~(\ref{LL8.37.1})) and convergence of the sequence $\rho_{n}(x)$ to
$\rho(x)$ in $L^{1}$ norm Eq.~(\ref{LL8.37.2}). Then it follows the
convergence in $L^{r}(\Omega,dm)$ distance for all values of $r$ in the range
$\min~(p,1)<r<\max~(p,1)$
\begin{equation}
D_{r,m}[\rho_{n}(x)-\rho(x)]\underset{n\rightarrow\infty}{\rightarrow}0
\label{LL8.37.3.2.3}%
\end{equation}
In particular we have for all the Tsallis and R\'{e}nyi entropies and
functionals $D_{r,m}[]$ with $r$ in this range
\begin{align}
&  S_{T,r}[\rho_{n}]\underset{n\rightarrow\infty}{\rightarrow}S_{T,r}%
[\rho]\label{LL8.37.4}\\
&  S_{R,r}[\rho_{n}]\underset{n\rightarrow\infty}{\rightarrow}S_{R,r}%
[\rho]\label{LL8.37.5}\\
&  D_{r,m}[\rho_{n}(x)]\underset{n\rightarrow\infty}{\rightarrow}D_{r,m}%
[\rho(x)]
\end{align}

\end{proposition}

\begin{remark}
From the previous result we note that the boundednes of the RE for some
$p\neq1$ has a stabilizing effect, only in the subset of all the PDFs
satisfying Eq.(\ref{LL8.37.1}). Hence, from the convergence in the natural
$L^{1}$ distance we can conclude the convergence of the R\'{e}nyi or Tsallis
entropies. Note that from Eq.(\ref{LL8.37.3}) and $p>1$, when $r\nearrow p$
the error in the $L^{r}$ norm increases and finally, when $r=p$ it attains the
upper bound Eq.(\ref{LL8.37.1}).
\end{remark}

From the point of view of statistical physics the convergence of sequence of
PDF defined by some $L^{p}$ distance is important because it follows the
convergence of the expectation values of physical observable that belongs to
the dual spaces \cite{ReedSimon}. Hence, in the case $p>1$, from
Eq.(\ref{LL8.37.3.2.3}) we get the convergence $\left\Vert \rho_{n}%
-\rho\right\Vert _{r}\rightarrow0$ for $1<r<p$ and it follows the continuity
of the expectation values of observable $f(x)~$ from the dual space $f\in
L^{q}(\Omega,dm)$ where $1/q+1/r=1$
\begin{align}
&  {\int\limits_{\Omega}}\rho_{n}(x)f(x)dm(x)\underset{n\rightarrow\infty
}{\rightarrow}{\int\limits_{\Omega}}\rho(x)f(x)dm(x){\ }\label{LL8.35}\\
f  &  \in L^{q}(\Omega,dm);~q>\frac{p}{p-1}%
\end{align}
For $0<q<1$, it is possible that the dual space $S$ of $L^{q}(\Omega,dm)$ is
trivial. When $S$ is not trivial (e.g., the space of sequences $l^{q}$) and
$f\in S$), from $L^{p}$ we obtain again the continuity of corresponding
expectation values, similar to Eq.(\ref{LL8.35}).

\paragraph{Continuity of RTE's when the PDF is approximated in $L^{p}$ norm.}

Suppose that the \textit{quality of approximation} $\rho_{n}(x)$ of the true
PDF $\rho(x)$ is quantified in some $L^{p}$ norm with $p\neq1$. A convenient
choice in the applications is $p=2$. Suppose that $\rho(x),~\rho_{n}(x)\in
L^{1}(\Omega,m)\cap L^{p}(\Omega,m)$ and in analogy to the previous case we
have
\begin{align}
{\int\limits_{\Omega}}\left\vert \rho_{n}(x)-\rho(x)\right\vert ^{p}dm(x)  &
\leq\varepsilon_{n}\underset{n\rightarrow\infty}{\rightarrow}0\label{lp1}\\
{\int\limits_{\Omega}}\rho_{n}(x)dm(x)  &  =1 \label{lp1.1}%
\end{align}
From the normalization condition, for $\rho,\rho_{n}$ Eq.(\ref{LL1.1},
\ref{lp1.1}) we obtain
\begin{equation}
{\int\limits_{\Omega}}\left\vert \rho_{n}(x)-\rho(x)\right\vert dm(x)\leq2
\label{lp2}%
\end{equation}
By using the Lemma \ref{markerInterpolationLemma} with $a_{1}=2$ and
$a_{p}=\varepsilon_{n}$ and Eq.(\ref{LL4.2}) we obtain for all $r$ in the
domain
\begin{equation}
D=\{r|\min(1,p)<r<\max(1,p)\} \label{lp2.1}%
\end{equation}
the following bounds
\begin{align}
{\int\limits_{\Omega}}\left\vert \rho_{n}(x)-\rho(x)\right\vert ^{r}dm(x)  &
\leq2^{\frac{p-r}{p-1}}\varepsilon_{n}^{\frac{r-1}{p-1}}\underset
{n\rightarrow\infty}{\rightarrow}0\label{lp3}\\
\left[  D_{r,m}[\rho_{n}-\rho]\right]  ^{\frac{1}{i(p)}}  &  \leq2^{\frac
{p-r}{p-1}}\varepsilon_{n}^{\frac{r-1}{p-1}} \label{lp4}%
\end{align}
By using Eqs.(\ref{lp4}, \ref{LL4.2}, respectively, \ref{LL4.41}) is
\begin{align}
\left\vert D_{r,m}[\rho_{n}]-D_{r,m}[\rho]\right\vert  &  \leq\delta
_{n}\underset{n\rightarrow\infty}{\rightarrow}0\label{lp5}\\
\delta_{n}  &  :=\left[  2^{\frac{p-r}{p-1}}\varepsilon_{n}^{\frac{r-1}{p-1}
}\right]  ^{i(p)} \label{lp6}%
\end{align}
From the Eqs.(\ref{lp5}, \ref{LL5}-\ref{LL7}) and Remark
\ref{markerREMconvergenceEquivallence} we have the following stability results

\begin{proposition}
\label{markerPropLpConvL1bound}Under the conditions Eqs.(\ref{LL1.1},
\ref{lp1}, \ref{lp1.1}) for $r$ in \ the range $\min(1,p)<r<\max(1,p)$ we
have
\begin{align*}
&  S_{T,r}[\rho_{n}]\underset{n\rightarrow\infty}{\rightarrow}S_{T,r}[\rho]\\
&  S_{R,r}[\rho_{n}]\underset{n\rightarrow\infty}{\rightarrow}S_{R,r}[\rho]
\end{align*}

\end{proposition}

We recall that in the case when the measure $dm(x)$ is probabilistic, if
$p>q\geq1~$\ then $\left\Vert f\right\Vert _{L_{q}}<\left\Vert f\right\Vert
_{L_{p}}$. If the measure $dm(x)$ is atomic, then from the convergence or
boundednes in some $l_{p}$ norm results the convergence or boundednes in all
$l_{q}$ norm, with $q>p$, including $l_{\infty}$.

\paragraph{}

\subsubsection{\bigskip Boundednes and stability of the BSE}

Details are given in Appendix \ref{markerSubsectCalculSumeIntegrale}

\paragraph{Counter example 1, discrete distribution
\label{markerSubsectCounterxDiscrete}}

We shall provide some counter examples that prove that the previous
stabilizing conditions are not sufficient for the stability of the BSE. The
details can be found in Appendix \ref{markerAppendCounterexDiscrete}. Consider
now the problem of continuity of the classical entropy, in the simplest case
of the countable infinite probability space, with probabilities $\ \mathbf{p}%
:=\left\{  p_{1},...,p_{n},...\right\}  :=\left\{  p_{k}\right\}
_{k=1}^{\infty}$. In this case
\begin{equation}
S_{cl}[\mathbf{p}]=-\sum\limits_{k=1}^{\infty}p_{k}\log p_{k} \label{LL8.4}%
\end{equation}
Clearly
\begin{equation}
1=\sum\limits_{k=1}^{\infty}p_{k}=\left\Vert \mathbf{p}\right\Vert _{l^{1}}
\label{LL8.5}%
\end{equation}
and it is logical to consider that the \textit{distance} between two
probability laws $\mathbf{p}:=\left\{  p_{1},...,p_{n},...\right\}  $ and
$\mathbf{p}^{\prime}:=\left\{  p_{1}^{\prime},...,p_{n}^{\prime},...\right\}
$ is given by $l^{1}$ distance:
\[
\left\Vert \mathbf{p-p}^{\prime}\right\Vert _{l^{1}}:=\sum\limits_{k=1}%
^{\infty}|p_{k}-p_{k}^{\prime}|
\]
Consider now the sequence that is convergent in the $l^{1}$
\begin{equation}
\mathbf{p}^{(n)}:=\left\{  p_{k}^{(n)}\right\}  _{k=1}^{\infty} \label{LL8.6}%
\end{equation}
where
\begin{align}
p_{k}^{(n)}  &  :=\frac{1}{K_{n}}\frac{1}{(k+4)\left[  \log\left(  k+4\right)
\right]  ^{2+\frac{1}{n}}}~\label{LL8.7}\\
K_{n}  &  =\sum\limits_{k=1}^{\infty}\frac{1}{(k+4)\left[  \log\left(
k+4\right)  \right]  ^{2+\frac{1}{n}}}<\infty\label{LL8.8}%
\end{align}
Note that $\left\Vert \mathbf{p}^{(n)}\right\Vert _{l^{1}}=1$ and the limit of
the sequence $\mathbf{p}^{(n)}$ in the space $l^{1}$ is the probability
distribution given by $\mathbf{p}:=\left\{  p_{k}\right\}  _{k=1}^{\infty}$
where
\begin{align}
\mathbf{p}  &  :=\left\{  p_{k}\right\}  _{k=1}^{\infty}\label{LL8.9}\\
p_{k}  &  =\frac{1}{K}\frac{1}{(k+4)\left[  \log\left(  k+4\right)  \right]
^{2}}\label{LL8.10}\\
K  &  =\sum\limits_{k=1}^{\infty}\frac{1}{(k+4)\left[  \log\left(  k+4\right)
\right]  ^{2}}<\infty\label{LL8.11}%
\end{align}
We have also the following (general) inclusion:
\begin{equation}
\left\Vert \mathbf{p}^{(n)}\right\Vert _{l^{1}}=1\Rightarrow\mathbf{p}%
^{(n)}\in l^{q}~;q>1 \label{LL8.11.1}%
\end{equation}
Despite the sequence $\mathbf{p}^{(n)}$ is convergent in $l^{1}$:
\[
\left\Vert \mathbf{p}^{(n)}-\mathbf{p}\right\Vert \underset{n\rightarrow
\infty}{\rightarrow}0
\]
it is easy to prove that the classical entropy is divergent:
\begin{equation}
S_{cl}[\mathbf{p}^{(n)}]=\mathcal{O}(n);~n\rightarrow\infty\label{LL8.12}%
\end{equation}
Consequently, the functional $\mathbf{p\rightarrow}S_{cl}[\mathbf{p}]$ from
the space $l^{1}$ to $\mathbb{R}$ is not continuos, despite the classical
entropy has the Lesche stability property and we have stabilizing condition
from Eq.(\ref{LL8.11.1}) in the form of $l^{p}$ boundednes with $p>1$.

\paragraph{Counter example 2, continuos distribution with finite measure $m$.
\label{markerSubsectCounterxContinous}}

A second counter example is the following (see also \cite{CLTmaxent}, page 7)%
\begin{align}
\rho_{n}(x)  &  =\frac{M_{n}}{x\left(  \log\frac{1}{x}\right)  ^{\alpha}%
};~x\geq\frac{1}{n}\label{cex2.1}\\
\rho_{n}(x)  &  =0;~0<x<\frac{1}{n}\label{cex2.2}\\
\rho(x)  &  =\frac{M}{x\left(  \log\frac{1}{x}\right)  ^{\alpha}%
}\label{cex2.3}\\
1  &  <\alpha<2 \label{cex2.4}%
\end{align}
where $M_{n}$, $M$ are finite normalization constants. Note that (for further
details see Appendix \ref{markerAppendixCounterexContinous})
\begin{equation}
D_{p}(\rho_{n})<C_{p};~D_{p}(\rho)<C_{p} \label{cex2.5}%
\end{equation}
for some constants $C_{p}<\infty$, that does not depend on $n$. Consequently
$\rho_{n},\rho\in L^{1}\cap L^{p}$ for all $0<p\leq1$ and
\begin{equation}
D_{1}[\rho_{n}-\rho]=\left\Vert \rho_{n}-\rho\right\Vert _{L^{1}}%
\underset{n\rightarrow\infty}{\rightarrow}0 \label{cex2.6}%
\end{equation}
so we have a wide choice of stabilizing conditions (see Proposition
\ref{markerPropstionRenyiContinuity}) assuring that $S_{R,q}(\rho
_{n})\rightarrow S_{R,q}(\rho)$ for all $q$ in the range $\ p<q<1$.
Nevertheless, the classical entropy diverges:
\begin{equation}
S_{cl}[\mathbf{\rho}_{n}]=-\mathcal{O}(\log n)^{2-\alpha}\rightarrow
-\infty\label{cex2.8}%
\end{equation}

\subsubsection{Boundednes and stability of the BSE.}

At the first stage, let us suppose in the following that the true PDF
$\rho(x)$ is approximated by the sequence $\rho_{n}\in L^{p}(\Omega,dm)\cap
L^{q}(\Omega,dm)$ where $0<p<1$ and $q>1$ in the $L^{1}$ norm:
\begin{equation}
\int\limits_{\Omega}\left\vert \rho_{n}(x)-\rho(x)\right\vert
dm(x):=\varepsilon_{n}\underset{n\rightarrow\infty}{\rightarrow}0 \label{sh0}%
\end{equation}
By Theorem \ref{markerTheoremRieszThorinGen} we have also $\rho_{n}\in
L^{1}(\Omega,dm)$. Suppose, without loss of generality, that in addition we
have the following stabilization conditions
\begin{align}
Z_{p,m}[\rho_{n}]  &  =D_{p,m}[\rho_{n}]\leq A;~p<1\label{sh1}\\
Z_{p,m}[\rho]  &  =D_{p,m}[\rho]\leq A;~p<1\label{sh1.1}\\
Z_{q,m}[\rho_{n}]  &  =D_{q,m}[\rho_{n}]^{q}\leq A;~q>1\label{sh2}\\
Z_{q,m}[\rho]  &  =D_{q,m}[\rho]^{q}\leq A;~q>1 \label{sh2.1}%
\end{align}
For the sake of simplicity we will used a less strict bounds. These bounds are
compatible with the normalization condition $Z_{1,m}[\rho_{n}]=1$ \ if $1\leq
A$. From Theorem \ref{markerTheoremRieszThorinGen}, from the Appendix
\ref{markerAppndxLogConvexity} and previous bounds, we get
\begin{align}
Z_{r,m}[\rho_{n}]  &  =\left[  D_{r,m}[\rho_{n}]\right]  ^{1/i(r)}\leq
A;~p\leq r\leq q\label{sh2.11}\\
Z_{r,m}[\rho]  &  =\left[  D_{r,m}[\rho]\right]  ^{1/i(r)}\leq A;~p\leq r\leq
q \label{sh2.12}%
\end{align}
From Eqs.(\ref{sh2.11}, \ref{sh2.12}, \ref{LL4.3}) we extend these bounds to
the complex domain, $p\leq\operatorname{Re}(w)\leq q$:
\begin{equation}
\left\vert Z_{w,m}[\rho_{n}]\right\vert =\left\vert {\int\limits_{\Omega}%
}\left\vert \rho_{n}(x)\right\vert ^{w}dm(x)\right\vert \leq{\int
\limits_{\Omega}}\left\vert \rho_{n}\right\vert ^{\operatorname{Re}%
(w)}dm(x)\leq A \label{sh2.2}%
\end{equation}
and in similar manner
\begin{equation}
\left\vert Z_{w,m}[\rho]\right\vert \leq A;~p\leq\operatorname{Re}(w)\leq q
\label{sh2.3}%
\end{equation}

\paragraph{Boundednes}

In the following $Z_{q,m}[\rho]$ can be considered a particular value
\cite{JizbaArimitsuAnalitic} of the analytic function $w\rightarrow
Z_{w,m}[\rho]$, and related to the BSE by Eq.(\ref{LL71.})

\begin{proposition}
\label{markerPropositionAnaliticityBound} Under the conditions Eqs.[\ref{sh1}%
,\ref{sh2} ] there exists an analytic continuation of the function
$Z_{r,m}[\rho_{n}]$, denoted by
\begin{equation}
F_{n}(z):=\int_{\Omega}\left\vert \rho_{n}(x)\right\vert ^{z}dm(x)
\label{sh3.01}%
\end{equation}
in the strip from the complex $z$ plane $D=\{z|p<\operatorname{Re}(z)<q\}$
such that $F_{n}(x)=Z_{x,m}[\rho_{n}]$ for $p\leq x\leq q$. We have for all
$\ z\in D$ the bound
\begin{equation}
\left\vert F_{n}(z)\right\vert \leq A \label{Sh3.1}%
\end{equation}
The BSE of the PDF is given by the Cauchy integral
\begin{equation}
S_{cl}[\rho_{n}]=-\left[  \frac{d}{dz}F_{n}(z)\right]  _{z=1}=-\frac{1}{2\pi
i}\oint\limits_{C}\frac{F_{n}(w)dw}{(w-1)^{2}} \label{sh4}%
\end{equation}
where $C$ is a sufficiently small circle centered in $w=1$. The BSE is bounded
by
\begin{equation}
\left\vert S_{cl}[\rho_{n}\right\vert \leq\frac{2A}{\min(q-1,1-p)}
\label{sh4.01}%
\end{equation}

\end{proposition}

\begin{proof}
The Eq.(\ref{Sh3.1}) results from Corollary
\ref{markeerCorolaryApendAnalitictyinStrip}, Appendix
(\ref{markerAppndxLogConvexity}). Eq.(\ref{sh4}) is a direct consequence of
Eqs~(\ref{LL71.}, \ref{sh3.01}). To prove Eq.(\ref{sh4.01}), we denote
\begin{equation}
r=\min(q-1,1-p)/2 \label{shp1}%
\end{equation}
the radius of the circle $C$, which is in the interior to the analyticity
domain of $F_{n}(w)$. By Cauchy theorem and Eq.(\ref{sh4}) we have
\begin{align}
F_{n}(z)  &  =\frac{1}{2\pi i}\oint\limits_{C}\frac{F_{n}(z^{\prime}
)}{z^{\prime}-z}dz^{\prime}\label{shp2}\\
S_{cl}[\rho_{n}]  &  =-\frac{1}{2\pi i}\oint\limits_{C}\frac{F_{n}(z^{\prime
})}{(z^{\prime}-1)^{2}}dz^{\prime} \label{shp3}%
\end{align}
where $C$ is the circle with radius $r$, so $|z^{\prime}-1|=r.$ By
Eqs.(\ref{sh2.2}, \ref{sh2.3})
\begin{equation}
\left\vert F_{n}(z^{\prime})\right\vert \leq A \label{shp3.1}%
\end{equation}
From Eq.(\ref{shp3}) we get
\[
|S_{cl}[\rho_{n}]|\leq\frac{1}{r}\underset{z\in C^{\prime}}{\max}|F_{n}(z)|
\]
which combined with Eqs.(\ref{shp1}, \ref{shp3.1}) completes the proof.
\end{proof}

\paragraph{Stability of the BSE}

.

\noindent Consider now the stability problem. Suppose that the sequence of PDF
$\rho_{n}$ approximates the exact PDF $\rho$ in the sense of Eq.(\ref{sh0})
and we have the bounds Eqs.(\ref{sh1}-\ref{sh2.1}). We obtain the following
result on the stability of the BSE:

\begin{theorem}
\label{markerTheoremShannonStability}Under the he previous conditions
\ Eqs.(\ref{sh0}-\ref{sh2.1}) we have%
\begin{equation}
\underset{n\rightarrow\infty}{\lim}S_{cl}[\rho_{n}]=S_{cl}[\rho] \label{sh6}%
\end{equation}

\end{theorem}

\noindent The proof (See Appendix
\ref{markerSubsectAppendShannonStabilityProof}) uses mathematical methods
adopted in high energy physics \cite{CiulliStabProblems}%
-\cite{NenciuHamonicMeasure}.

\noindent In conclusion, the problem of numerical approximation of the
classical and generalized entropies in the general case, cannot be solved
without additional assumptions, that are tacitly used in practice. Among
auxiliary assumptions that could stabilize the numerical instability in the
computation of the entropies, we have smoothness conditions and bounds on the
tail of the probability density function. We also mention that in many
practical situation, the cumulative probability distribution function is
approximated by the empirical cumulative distribution obtained from
experiment. In this case the goodness of the fitting is characterized by a
random variable, whose distribution is described by the Kolmogorov-Smirnov or
the Anderson-Darling statistics. Consequently, the computed entropies are
itself random variables and the problem of continuity must be formulated in
terms of convergence of random variables. This class of problems deserves a
further study.

\section{The generalized R\'{e}nyi entropies (GRE). \label{markerSectionGRE}}

\subsection{Motivations}

Our generalization of the GRE defined in the previous work
\cite{SonninoSteinbrGRE} is a straightforward extension of the previous case
when the full phase space is a Cartesian product of two spaces, to the more
general case of $N$ factors. The line of reasoning is the same: first, we
remark the mathematical relation between classical R\'{e}nyi entropy and the
metric in Lebesgue spaces and, second, we define the new entropies by using
the metric in generalized Lebesgue spaces \cite{Besov} (similar to the
particular case $N=2$ studied in \cite{SonninoSteinbrGRE}). For
characterization of generalized entropies from the point of view of
fundamental mathematical structures, see \cite{SGySASGCategory}.

\subsubsection{Avoiding integrability problems}

The PDF functions in many variables may have a more complex singularity
structure such that all the RTE diverge. We proved in \cite{SonninoSteinbrGRE}
that the GRE can solve this problem, in the case of PDF with two variables.
The generalization introduced here is intended to treat similar integrability
problems in many variables. As specified in the previous part, the domain of
the entropy parameter $q$, where the R\'{e}nyi or the Tsallis entropies are
defined, is related to the singularities and to the asymptotic behavior of the
PDF. Indeed, consider the case when the PDF depends on the real variable $x$
so $(\Omega,\mathcal{A},m)=(\mathbb{R},\mathcal{B}(\mathbb{R}),\lambda_{1})$
where $\mathcal{B}(\mathbb{R})$ is the family of subsets of $\mathbb{R}$
generated by denumerable intersections and unions of open and close intervals,
and $\lambda_{1}$ is the Lebesgue measure (more exactly $dm=d\lambda_{1}=dx$).
Suppose that the PDF $\rho(x)$ has the decomposition in a regular part,
$\rho_{0}(x)$ is differentiable and decays, at least exponentially at
infinity, in both singular part and a heavy tail part
\[
\rho(x)=\rho_{0}(x)\left[  1+\frac{A}{\left\vert x\right\vert ^{\alpha}%
}\right]  +\frac{B}{1+\left\vert x\right\vert ^{1+\beta}}%
\]
where $A,B,\beta~>0$, and $0<\alpha<1$. The RTE is finite when $1/\left(
1+\beta\right)  <q<1/\alpha$. In processing the experimental data, the large
value of the RTE for some $q<1$ is a signal characterized by a heavy tail of
PDF (this is shown in the stationary PDF in self organized criticality models
as well as in the linearized stochastic models with multiplicative noise
\cite{sgbw}). The large value of RTE for some $q>1$ suggests the existence of
local singularity, that in some models is related to the phenomenon of noise
driven intermittency \cite{intermittency}. Consequently, the plotting of the
RTE is a good practice for detecting the existence of singularity and heavy
tail effects.

\noindent Suppose now that we have a PDF depending of $3$ variables
$\rho(x_{1} ,x_{2},x_{3})$ that similarly has a regular part $\rho_{0}$ that
is differentiable and decay exponentially at infinity, and singular and heavy
tail parts
\begin{equation}
\rho(x_{1},x_{2},x_{3})=\rho_{0}(x_{1},x_{2},x_{3})\left[  1+\sum
\limits_{i=1}^{3}\frac{A_{i}}{\left\vert x_{i}\right\vert ^{\alpha_{i}}
}\right]  +\prod\limits_{i=1}^{3}\frac{B_{i}}{1+\left\vert x_{i}\right\vert
^{1+\beta_{i}}} \label{LL8.13}%
\end{equation}
\noindent with $\alpha_{1}<\alpha_{2}<\alpha_{3}$ and $\beta_{1}<\beta
_{2}<\beta_{3}$. Note that the RTE associated to $\rho$ from Eq.(\ref{LL8.13})
is defined in the domain
\begin{equation}
1/\left(  1+\beta_{1}\right)  <q<1/\alpha_{3} \label{LL8.14}%
\end{equation}
so it cannot detect the singularity (possible intermittency phenomena) in the
variables $x_{1},x_{2}$ and the heavy tail (possible SOC effects) in the
variable $x_{2},x_{3}$. In the sequel we will introduce a subsequent
generalization that overcomes this difficulty.

\subsubsection{Invariance under measure preserving transformations.}

The classical as well as generalized entropies are general functionals that
extracts/condense information about PDF, according to most general rules of
the probability theory. As a result the BSE, RTE are invariant under measure
preserving transformations. For illustration we consider the simplest case,
when we have a finite discrete probability distribution $p_{i,\alpha}$, with
$1\leq i\leq N$ and $1\leq\alpha\leq A$, with $\sum\limits_{i=1}^{N}%
\sum\limits_{a=1}^{A}p_{i,\alpha}=1$. Consider the case when the measure $m$
in Eqs(\ref{LL4.3}, \ref{LL5}) is the product of the counting measures on the
sets $\{1,...,N\}$ \ and $\{1,...,A\}$. In the terminology of "Big Data", the
$NA$ data are presented in tensorized form \cite{BigData1}. The corresponding
BSE and RE are
\begin{align}
S_{cl}  &  =-\sum\limits_{i=1}^{N}\sum\limits_{a=1}^{A}p_{i,\alpha}\log
p_{i,\alpha}\label{LL8.141}\\
S_{R,w}  &  =\frac{1}{1-w}\log\sum\limits_{i=1}^{N}\sum\limits_{a=1}%
^{A}p_{i,\alpha}^{w};~w>0 \label{ll8.142}%
\end{align}
These entropies are invariant under the change of variables
\[
p_{i,\alpha}\rightarrow p_{\sigma(i,\alpha)}%
\]
where the transformation $\sigma:$ $\ (i,\alpha)\rightarrow\sigma(i,\alpha)$
is an element of the permutation group $\mathcal{S}_{NA}$ of $NA$ objects,
\ $\mathcal{S}_{NA}$ having $(NA)!$ elements, that reflect our complete lack
of information about state space and ignore its Cartesian product structure.
However, when the indices $i$ and $\alpha$ have different physical meaning,
such an extended symmetry hypothesis is not appropriate. On the other hand,
the GRE's whose construction use the Cartesian product structure for generic
$p_{i,\alpha}$ \cite{SonninoSteinbrGRE}
\begin{align}
S_{v,w}  &  :=\frac{1}{1-w}\log\sum\limits_{i=1}^{N}\left[  \sum
\limits_{a=1}^{A}p_{i,\alpha}^{w}\right]  ^{v}\label{LL8.145}\\
S_{p,q}^{(permuted)}  &  :=\frac{1}{1-w}\log\sum\limits_{\alpha=1}^{A}\left[
\sum\limits_{i=1}^{N}p_{i,\alpha}^{w}\right]  ^{v} \label{LL8.145a}%
\end{align}
is not invariant under full permutation group $\mathcal{S}_{NA}$. The
knowledge of $S_{R,q}$ from Eq.(\ref{ll8.142}) for all $w>0$ allows to
reconstruct the probabilities $p_{i,a}$ modulo permutation group
$\mathcal{S}_{NA}$ (i.e. reconstruct the probabilities without specification
of their place in the list). On the other hand the knowledge of the GRE's from
Eqs.(\ref{LL8.145}, \ref{LL8.145a}) for all $v>0,~w>0$ allows to reconstruct
$p_{i,a}$ modulo smaller group $\mathcal{S}_{N}\times\mathcal{S}_{A}$ (see
\cite{SGySASGCategory}). Similar problems appear in the case of probability
distribution $p_{i,\alpha,m}$ where the indices $i,\alpha,m$ have different
interpretation. The invariance group $\mathcal{S}_{NAM}$ of the R\'{e}nyi
entropy
\[
S_{R,q}=\frac{1}{1-q}\log\sum\limits_{i=1}^{N}\sum\limits_{a=1}^{A}%
\sum\limits_{m=1}^{M}p_{i,\alpha,m}^{q}%
\]
is too large, it contains $(NAM)!$ elements. In contrast, the generalized
R\'{e}nyi entropy, introduced in this work
\begin{equation}
\frac{1}{1-q_{3}}\log\sum\limits_{i=1}^{N}\left[  \sum\limits_{a=1}^{A}\left[
\sum\limits_{m=1}^{M}p_{i,\alpha,m}^{q_{3}}\right]  ^{q_{2}}\right]  ^{q_{1}}
\label{LL8.148}%
\end{equation}
is no more invariant under full permutation group $\mathcal{S}_{NAM}$. On the
other hand it is easy to see that the invariance group of Eq.(\ref{LL8.148})
contains at least the product of subgroups, having at least $N!A!M!$ elements.
For more detailed discussion, see \cite{SGySASGCategory}. In conclusion, while
the RTE, BSE are constructed by using the most fundamental structures of the
probability theory, the GRE also take into account the Cartesian product
structure of the phase space and the corresponding product structures of measures.

\subsection{Definitions and notations \label{markerSubsectGREdeffinition}}

We follow the same approach as in ref.\cite{SonninoSteinbrGRE}. We will define
the Generalized R\'{e}nyi entropies by using the results on Banach spaces with
the anisotropic norm, exposed in ref.\cite{Besov}. In the first part of the
discussion we will restrict our discussions to the set of parameters that
defines the GRE, when a) The integrals that appear in the definition can be
interpreted as the distance in a suitable function space and b) The formula
for entropy can be related to convexity or concavity properties of some
functional, in the subspace of non negative density functions. Consequently,
we will define only two class of distance-functionals and entropies, in
analogy to functionals $S_{p_{y},p_{z}}^{(1)}[\rho]$ and $S_{q_{y},q_{z}%
}^{(2)}[\rho]$ defined in ref.\cite{SonninoSteinbrGRE}.

Consider that the measure space $(\Omega,\mathcal{A},m)$ has the following
direct product structure. First, the phase space $\Omega$ is split into $N$
subspaces
\begin{equation}
\Omega=\Omega_{1}\times\Omega_{2}\times...\times\Omega_{N} \label{LL9}%
\end{equation}
This means that the argument $\mathbf{x}$ of probability density function can
be represented as $\mathbf{x=}\left\{  x_{1},x_{2},...,x_{N}\right\}  $, so
\begin{equation}
\rho(\mathbf{x})=\rho(x_{1},x_{2},...,x_{N}) \label{LL9.1}%
\end{equation}
with $x_{k}\in\Omega_{k}$, $1\leq k\leq N$. We mention also that, in general,
it is possible that the component spaces $\Omega_{k}$ are discrete, or has the
structure of $\mathbf{R}^{n}$ or more in general, of an infinite dimensional
measure space. Each of the spaces $\Omega_{k}$ has their $\sigma-$algebra
$\mathcal{A}_{k}$. (The $\sigma-$algebra $\mathcal{A}$, that contains subsets
of $\Omega=\Omega_{1}\times\Omega_{2}\times...\times\Omega_{N}$ is defined as
a tensor product of the $\sigma-$algebras $\mathcal{A}_{k}$ : it is the
smallest $\sigma-$algebra on $\Omega$ such that all of the projections
$\Omega\overset{p_{k}}{\rightarrow}\Omega_{k}$ are measurable \cite{Rudin}%
-\cite{EStein}). The measure $m$ is also factorizable:
\begin{equation}
dm(\mathbf{x)=}dm(x_{1},x_{2},...,x_{N}\mathbf{\mathbf{)}=}\prod
\limits_{j=1}^{N}dm_{k}(x_{k}) \label{LL10}%
\end{equation}
where the measures $m_{k}$ are defined on the $\sigma-$algebras (space of
events) $\mathcal{A}_{k}$. In other words, the measure space $(\Omega
,\mathcal{A},m)$ is the direct product
\begin{equation}
(\Omega,\mathcal{A},m)=\bigotimes\limits_{j=1}^{N}(\Omega_{j},\mathcal{A}%
_{j},m_{j}) \label{LL11}%
\end{equation}
The elementary probability $dP(\mathbf{x})$ is given by
\begin{equation}
dP(\mathbf{x})=\rho(x_{1},x_{2},...,x_{N})dm(\mathbf{x}) \label{LL12}%
\end{equation}
where $dm(\mathbf{x})$ is given by Eq.(\ref{LL10}). Consider a vector
$\mathbf{p}=\{p_{1},p_{2},...,p_{N}\}$ of real numbers with $p_{k}\geq1$.
According to Ref.\cite{Besov}, in close analogy to
Ref.\cite{SonninoSteinbrGRE} (where the particular case $N=2$ was studied), we
define recursively the anisotropic norm (depending on the measure $m$)
$\left\Vert \rho\right\Vert _{\mathbf{p},m}$ as follows (see Appendix,
subsection \ref{markerSubsectionPropertiesOfFunctionalDpmf})
\begin{align}
\rho_{N}(,x_{1},x_{2},...,x_{N})  &  :=\rho(x_{1},x_{2},...,x_{N}%
)\label{LL13}\\
\rho_{N-1}(x_{1,},x_{2},...,x_{N-1})  &  :=\left[  \int\limits_{\Omega_{N}%
}\left[  \rho_{N\text{ }}(x_{1},...,x_{N})\right]  ^{p_{N}}dm_{N}%
(x_{N})\right]  ^{1/p_{N}}....\label{LL14}\\
\rho_{k-1}(x_{1},x_{2},...,x_{k-1})  &  :=\left[  \int\limits_{\Omega_{k}%
}\left[  \rho_{k\text{ }}(x_{1},...,x_{k})\right]  ^{p_{k}}dm_{k}%
(x_{k})\right]  ^{1/p_{k}}...\label{LL15}\\
\rho_{1\text{ }}(x_{1})  &  :=\left[  \int\limits_{\Omega_{2}}\left[
\rho_{2\text{ }}(x_{1},x_{2})\right]  ^{p_{2}}dm_{2}(x_{2})\right]  ^{1/p_{2}%
}\ \label{LL16}\\
\left\Vert \rho\right\Vert _{\mathbf{p},m}  &  :=\left[  \int\limits_{\Omega
_{1}}\left[  \rho_{1\text{ }}(x_{1})\right]  ^{p_{1}}dm_{1}(x_{1})\right]
^{1/p_{1}}\ \label{LL17}%
\end{align}
In analogy with Eqs.(\ref{LL3}, \ref{LL5}) and Ref.~\cite{SonninoSteinbrGRE}
we define the GRE, with respect to the measure $m$
\begin{equation}
S_{\mathbf{p}}^{(1)}[\rho,m]=\frac{p_{1}}{1-p_{N}}\log\left\Vert
\rho\right\Vert _{\mathbf{p},m};p_{i}>1 \label{LL18}%
\end{equation}
Note that the also the anisotropic norm function $\rho\rightarrow\left\Vert
\rho\right\Vert _{\mathbf{p},m}$ is convex and satisfies the axioms related to
the norm~\cite{Besov}. The corresponding normed vector space is complete, i.
e. it is a Banach space (see Ref.\cite{Besov}). There is another range of
parameters that generalizes the R\'{e}nyi entropy corresponding to
Eqs.(\ref{LL4}, \ref{LL6}). Consider a vector $\mathbf{q}=\{q_{1}%
,q_{2},...,q_{N}\}$ of real numbers with $0<q_{k}\leq1$. In analogy to
Eqs.(\ref{LL13}-\ref{LL17}) we define recursively
\begin{align}
\rho_{N}^{\prime}(x_{1},x_{2},...,x_{N})  &  :=\rho(x_{1},x_{2},...,x_{N}%
)\label{LL20}\\
\rho_{N-1}^{\prime}(x_{1},x_{2},...,x_{N-1})  &  :=\int\limits_{\Omega_{N}%
}\left[  \rho_{M\text{ }}^{\prime}(x_{1},...,x_{N})\right]  ^{q_{N}}%
dm_{N}(x_{N})....\label{LL21}\\
\rho_{k-1}^{\prime}(x_{1},x_{2},...,x_{k-1})  &  :=\int\limits_{\Omega_{k}%
}\left[  \rho_{k\text{ }}^{\prime}(x_{1},...,x_{k})\right]  ^{q_{k}}%
dm_{k}(x_{k})...\label{LL23}\\
\rho_{1\text{ }}^{\prime}(x_{1})  &  :=\int\limits_{\Omega_{2}}\left[
\rho_{2\text{ }}^{\prime}(x_{1},x_{2})\right]  ^{q_{2}}dm_{2}(x_{2}%
)\label{LL24}\\
N_{\mathbf{q},m}[\rho]  &  :=\int\limits_{\Omega_{1}}\left[  \rho_{1\text{ }%
}^{\prime}(x_{1})\right]  ^{q_{1}}dm_{1}(x_{1}) \label{LL25}%
\end{align}
Observe that the mapping $\rho\rightarrow N_{\mathbf{q},m}[\rho]$ defines a
pseudo-norm on the space of probability density functions (see Appendix,
subsection \ref{markerSubsectionPropertiesOfFunctionalDpmf}). The map
$\rho\rightarrow N_{\mathbf{q},m}[\rho]$ defines a concave function, in the
subset of physically admissible PDF's, when $\rho(x_{1},x_{2},...,x_{N})\geq
0$. The GRE will be defined in analogy to Eqs.(\ref{LL4}, \ref{LL6}) and to
the case $N=2$ from Ref.~\cite{SonninoSteinbrGRE}%
\begin{equation}
S_{\mathbf{q}}^{(2)}[\rho,m]=\frac{1}{1-q_{N}}\log N_{\mathbf{q},m}%
[\rho];~0<q_{i}<1 \label{LL26}%
\end{equation}
For the sake of simplicity, we will use the extrapolated form of
Eq.~(\ref{LL26}) in all range $q_{k}>0$ of \ the parameters $\{q_{1}%
,...,q_{N}\}$ allowing to relate $S_{\mathbf{q}}^{(2)}[\rho,m]$ to
$S_{\mathbf{p}}^{(1)}[\rho,m]$. We obtain
\begin{align}
S_{\mathbf{q}}^{(2)}[\rho,m]  &  =S_{\mathbf{p}}^{(1)}[\rho,m]\label{LL27}\\
N_{\mathbf{q},m}[\rho]  &  =\left[  \left\Vert \rho\right\Vert _{\mathbf{p}%
,m}\right]  ^{p_{1}} \label{LL28}%
\end{align}
when $p_{i}$ and $q_{i}$ are related as follows
\begin{align}
q_{N}  &  =p_{N}\label{LL29}\\
q_{N-1}  &  =\frac{p_{N-1}}{p_{N}}...\label{LL30}\\
q_{k}  &  =\frac{p_{k}}{p_{k+1}},...\label{LL31}\\
q_{1}  &  =\frac{p_{1}}{p_{2}} \label{LL32}%
\end{align}

\begin{remark}
\label{RemDomain_p_q} The algebraic equations for the Lagrange multipliers
associated to maximal entropy problem are very complicated in the general
case, nevertheless from the convexity or concavity properties the uniqueness
of the solution follows. According to Eqs.(\ref{LL27}-\ref{LL32}), we are in
the domain when $\rho\rightarrow\left\Vert \rho\right\Vert _{\mathbf{p},m}$ is
a convex functional when
\begin{equation}
p_{k}=\prod\limits_{j=k}^{N}q_{j}\geq1 \label{LL33}%
\end{equation}
In this case, the problem of maximal entropy with linear restriction is
equivalent to minimization of a positive convex function and has unique
solution. In the domain $0<q_{k}<1$, where the map $\rho\rightarrow
N[\rho]_{\mathbf{q,}m}$ is a concave function, the generalized MaxEnt problem
is equivalent to the maximization of a concave function with linear
restriction. If the solution exists, it is unique. In the more general case
exposed below the problem of uniqueness deserves further study.
\end{remark}

To have a more compact, more general extension of both the definitions of
$\left\Vert \rho\right\Vert _{\mathbf{p},m}$ in Eqs.(\ref{LL13} -\ref{LL17})
and $N_{\mathbf{q},m}[\rho]$ in Eqs.(\ref{LL20} -\ref{LL25}), we define a more
general distance functional for a more general range of parameters
$\mathbf{p}$, that are generalizations of the functional defined in
Eq.(\ref{LL4.2}). We shall use the notation of Eq.(\ref{LL4.1}), so the
generalization of Eqs.(\ref{LL13}-\ref{LL17}, \ref{LL20}-\ref{LL25}) are the
following (for another equivalent construction see Appendix subsection
\ref{markerSubsectionPropertiesOfFunctionalDpmf})
\begin{align}
f_{N}(,x_{1},x_{2},...,x_{N})  &  :=f(x_{1},x_{2},...,x_{N})\label{33.1}\\
f_{N-1}(x_{1,},x_{2},...,x_{N-1})  &  :=\left[  \int\limits_{\Omega_{N}
}\left\vert f_{N}(x_{1},...,x_{N})\right\vert ^{p_{N}}dm_{N}(x_{N})\right]
^{i(p_{N})}....\label{33.2}\\
f_{k-1}(x_{1},x_{2},...,x_{k-1})  &  :=\left[  \int\limits_{\Omega_{k}
}\left\vert f_{k\text{ }}(x_{1},...,x_{k})\right\vert ^{p_{k}}dm_{k}
(x_{k})\right]  ^{i(p_{k})}...\label{33.3}\\
f_{1}(x_{1})  &  :=\left[  \int\limits_{\Omega_{2}}\left\vert f_{2\text{ }
}(x_{1},x_{2})\right\vert ^{p_{2}}dm_{2}(x_{2})\right]  ^{i(p_{2}
)}\ \label{33.4}\\
D_{\mathbf{p},m}\left[  f\right]   &  :=\left[  \int\limits_{\Omega_{1}
}\left\vert f_{1\text{ }}(x_{1})\right\vert ^{p_{1}}dm_{1}(x_{1})\right]
^{i(p_{1})}\ \label{33.5}%
\end{align}
We have the following important properties

\begin{proposition}
\label{markerPropositionD[f]Pseudonorm}The functional $f\rightarrow
D_{\mathbf{p},m}\left[  f\right]  $ is a pseudo norm
\begin{align}
D_{\mathbf{p},m}\left[  \alpha f\right]   &  =\left\vert \alpha\right\vert
^{s}D_{\mathbf{p},m}\left[  f\right]  ;~\alpha\in\mathbb{R}\label{33.6}\\
D_{\mathbf{p},m}\left[  f+g\right]   &  \leq D_{\mathbf{p},m}\left[  f\right]
+D_{\mathbf{p},m}\left[  g\right] \label{33.7}\\
\left\vert D_{\mathbf{p},m}\left[  f\right]  -D_{\mathbf{p},m}\left[
g\right]  \right\vert  &  \leq D_{\mathbf{p},m}\left[  f-g\right]
\label{33.8}%
\end{align}
where the homogeneity degree from Eq.(\ref{33.6}) \ is~ $s=\prod
\limits_{k=1}^{N}\left[  i(p_{k})p_{k}\right]  $
\end{proposition}

For proof see Appendix, Subsection
\ref{markerSubsectionPropertiesOfFunctionalDpmf}

\noindent Always is it possible to relate the functionals $D_{\mathbf{p}%
,m}\left[  f\right]  $ by $N_{\mathbf{q},m}\left[  f\right]  $. By comparing
\ Eqs.(\ref{LL20}-\ref{LL25}) we obtain
\begin{equation}
D_{\mathbf{p},m}\left[  f\right]  =\left[  N_{\mathbf{q},m}[f]\right]
^{i(p_{1})} \label{33.9}%
\end{equation}
where the relation between exponents $\mathbf{q=}(q_{1},...,q_{N})$ and
$\mathbf{p}=(p_{1},...,p_{N})$ is
\begin{align}
q_{N}  &  =p_{N}\label{33.10}\\
q_{N-1}  &  =p_{N-1}i(p_{N})\label{33.11}\\
q_{N-2}  &  =p_{N-2}i(p_{N-1})\label{33.12}\\
&  ...\\
q_{2}  &  =p_{2}i(p_{3})\label{33.13}\\
q_{1}  &  =p_{1}i(p_{2}) \label{33.16}%
\end{align}
From previous equations we have that the GRE can be expressed always either by
functional $D_{\mathbf{p},m}\left[  \rho\right]  $ or by the functional
$N_{\mathbf{q},m}[\rho]$
\begin{equation}
S_{\mathbf{q}}^{(2)}[\rho,m]=\frac{1}{1-q_{N}}\log N_{\mathbf{q},m}%
[\rho]=S_{\mathbf{q}}^{(2)}[\rho,m]=\frac{1}{i(p_{1)}(1-p_{N})}\log
D_{\mathbf{p},m}\left[  f\right]  \label{33.17}%
\end{equation}

\begin{remark}
Beside the previous definitions a more complete information about PDF can be
obtained from $S_{\mathbf{q}}^{(2)}[\rho^{(perm)},m]$ where ~$\rho
^{(perm)}(x_{1},...,x_{N}):=\rho(x_{T(1)},...,x_{T(N)})$ and the map
$k\rightarrow T(k)$ is a permutation of $N$ indices.
\end{remark}

\subsection{Properties of the GRE}

\subsubsection{Extensivity, in the classical sense}

The extensivity follows from the multiplicative property of the norms
$\left\Vert \rho\right\Vert _{\mathbf{p}}$ or pseudo-norms $N[\rho
]_{\mathbf{q}}$. Suppose that for all of measure spaces $(\Omega
_{j},\mathcal{A}_{j},m_{j})$ that appear in Eq.(\ref{LL11}) we have the
splitting
\begin{equation}
(\Omega_{j},\mathcal{A}_{j},m_{j})=(\Omega_{j}^{(1)},\mathcal{A}_{j}%
^{(1)},m_{j}^{(1)})\otimes(\Omega_{j}^{(2)},\mathcal{A}_{j}^{(2)},m_{j}^{(2)})
\label{LL34}%
\end{equation}
that means in particular that $\Omega_{j}=\Omega_{j}^{(1)}\times\Omega
_{j}^{(2)}$, $x_{j}=\{x_{j}^{(1)},x_{j}^{(2)}\}\in\Omega_{j}$,
\begin{equation}
dm_{j}(x_{j})=dm_{j}(x_{j}^{(1)},x_{j}^{(2)})=dm_{j}^{(1)}(x_{j}^{(1)}%
)dm_{j}^{(2)}(x_{j}^{(2)}) \label{LL35}%
\end{equation}
Accordingly we have the splitting of the phase space $\Omega$
\begin{align}
\Omega &  =\Omega^{(1)}\times\Omega^{(2)}\label{LL36}\\
dm(\mathbf{x)}  &  \mathbf{=}dm^{(1)}(\mathbf{x}^{(1)}\mathbf{)}%
dm^{(2)}(\mathbf{x}^{(2)}\mathbf{)}\label{LL37}\\
\mathbf{x}^{(a)}  &  =\{x_{1}^{(a)},x_{2}^{(a)}...,x_{N}^{(a)}\};~a=\overline
{1,2} \label{LL38}%
\end{align}
where
\begin{align}
\Omega^{(a)}  &  =\Omega_{1}^{(a)}\times\Omega_{2}^{(a)}\times...\times
\Omega_{N}^{(a)};~a=\overline{1,2}\label{LL39}\\
dm^{(a)}(\mathbf{x}^{(a)}\mathbf{)}  &  \mathbf{=}\prod\limits_{k=1}^{N}%
dm_{j}^{(a)}(x_{j}^{(a)});~a=\overline{1,2} \label{LL40}%
\end{align}
or in compact notation
\begin{equation}
(\Omega,\mathcal{A},m)=(\Omega^{(1)},\mathcal{A}^{(1)},m^{(1)})\otimes
(\Omega^{(2)},\mathcal{A}^{(2)},m^{(2)}) \label{LL41}%
\end{equation}
Suppose that the PDF is also factorized
\begin{equation}
\rho(\mathbf{x})=\rho_{1}(\mathbf{x}^{(1)})\rho_{2}(\mathbf{x}^{(2)})
\label{LL42}%
\end{equation}
Then we have the following relations, for all values of the parameters
$\{p_{1},...,p_{N}\}$ or $\{q_{1},...,q_{N}\}\,\ $such that the integrals make
sense
\begin{align}
\left\Vert \rho\right\Vert _{\mathbf{p},m}  &  =\left\Vert \rho_{1}\right\Vert
_{\mathbf{p},m_{1}}\left\Vert \rho_{2}\right\Vert _{\mathbf{p},m_{2}%
}\label{LL43}\\
N_{\mathbf{q},m}[\rho]  &  =N_{\mathbf{q},m_{1}}[\rho_{1}]N_{\mathbf{q},m_{2}%
}[\rho_{2}]\label{LL44}\\
S_{\mathbf{p}}^{(1)}[\rho,m]  &  =S_{\mathbf{p}}^{(1)}[\rho_{1},m_{1}%
]+S_{\mathbf{p}}^{(1)}[\rho_{2},m_{2}]\label{LL45}\\
S_{\mathbf{q}}^{(2)}[\rho,m]  &  =S_{\mathbf{q}}^{(2)}[\rho_{1},m_{1}%
]+S_{\mathbf{q}}^{(2)}[\rho_{2},m_{2}] \label{LL46}%
\end{align}
In the previous relations we defined $S_{\mathbf{p}}^{(1)}[\rho,m]$ and
$S_{\mathbf{q}}^{(2)}[\rho,m]$ according to Eqs.(\ref{LL18}, \ref{LL26}) and
correspondingly
\begin{align}
S_{\mathbf{p}}^{(1)}[\rho_{a},m_{a}]  &  =\frac{p_{1}}{1-p_{N}}\log\left\Vert
\rho_{a}\right\Vert _{\mathbf{p},m_{a}};~\ a=\overline{1,2}\label{LL48}\\
S_{\mathbf{q}}^{(2)}[\rho_{a},m_{a}]  &  =\frac{1}{1-q_{N}}\log N_{\mathbf{q}%
,m_{a}}[\rho_{a}];~\ a=\overline{1,2} \label{LL49}%
\end{align}
For the sake of clarity, consider the following example with $N=3$. Suppose
$0<q_{1},q_{2},q_{3}\leq1$ and $p_{1},p_{2},p_{3}\geq1$. We define
$N_{\mathbf{q},m_{a}}[\rho_{a}]$ respectively $\left\Vert \rho_{a}\right\Vert
_{\mathbf{p},m_{a}}$ with $a=\overline{1,2}$ as follows
\begin{align}
N_{\mathbf{q},m_{a}}[\rho_{a}]  &  =\int\limits_{\Omega_{1}^{(a)}}dm_{1}%
^{(a)}(x_{1}^{(a)})\label{LL50}\\
&  \left[  \int\limits_{\Omega_{2}^{(a)}}dm_{2}^{(a)}(x_{2}^{(a)})\left[
\int\limits_{\Omega_{3}^{(a)}}dm_{3}^{(a)}(x_{3}^{(a)})\rho\left(  x_{1}%
^{(a)},x_{2}^{(a)},x_{3}^{(a)}\right)  ^{q_{3}}\right]  ^{q_{2}}\right]
^{q_{1}}\\
\left\Vert \rho_{a}\right\Vert _{\mathbf{p},m_{a}}^{p_{1}}  &  =\int
\limits_{\Omega_{1}^{(a)}}dm_{1}^{(a)}(x_{1}^{(a)})\label{LL51}\\
&  \left[  \int\limits_{\Omega_{2}^{(a)}}dm_{2}^{(a)}(x_{2}^{(a)})\left[
\int\limits_{\Omega_{3}^{(a)}}dm_{3}^{(a)}(x_{3}^{(a)})\rho\left(  x_{1}%
^{(a)},x_{2}^{(a)},x_{3}^{(a)}\right)  ^{p_{3}}\right]  ^{p_{2}/p_{3}}\right]
^{p_{1}/p_{2}}%
\end{align}

\subsubsection{Particular cases.}

In the following we will omit the measure, when no confusion arise:
$\left\Vert \rho\right\Vert _{\mathbf{p},m}:=\left\Vert \rho\right\Vert
_{\mathbf{p}}$; $N[\rho]_{\mathbf{q},m}:=N[\rho]_{\mathbf{q}}$;
\ $S_{\mathbf{p}}^{(a)}[\rho,m]:=S_{\mathbf{p}}^{(a)}[\rho]$. In the
particular case when $p_{1}=...=p_{N}>1$, or $0<q_{1}=q_{2}=...=q_{N-1}=1$ and
$\ q_{N}<1$ the GRE is equal to the classical R\'{e}nyi entropy from
Eqs.(\ref{LL5}, \ref{LL6}):
\begin{align}
S_{\mathbf{p}}^{(1)}[\rho]  &  =S_{R,p_{N}}[\rho]=\frac{p_{N}}{1-p_{N}}%
\log\left\Vert \rho\right\Vert _{p_{N}}\label{LL52}\\
\left\Vert \rho\right\Vert _{p_{N}}  &  =\left[  \int\limits_{\Omega
}dm(\mathbf{x})\rho(\mathbf{x})^{p_{N}}\right]  ^{1/p_{N}} \label{LL53}%
\end{align}
respectively
\begin{align}
S_{\mathbf{q}}^{(2)}[\rho]  &  =S_{R,q_{N}}[\rho]=\frac{1}{1-q_{N}}\log
N[\rho]_{q_{N}}\label{LL54}\\
N_{q_{N}}[\rho]  &  =\int\limits_{\Omega}dm(\mathbf{x})\rho(\mathbf{x}%
)^{q_{N}} \label{LL55}%
\end{align}
We used the notation from Eqs.(\ref{LL20}-\ref{LL26}). In particular when
$p_{N}\searrow1$ in Eqs.(\ref{LL52}, \ref{LL53}) and when $q_{1}%
=q_{2}=...=q_{N-1}=1;~q_{N}\nearrow1$ in Eqs.(\ref{LL54}, \ref{LL55}),
respectively, we obtain the classical BSE
\begin{equation}
\underset{p_{1}=...=p_{N}\searrow1}{\lim}S_{\mathbf{p}}^{(1)}[\rho
]=\underset{q_{1}=...=q_{N-1}=1;~q_{N}\nearrow1}{\lim}S_{\mathbf{q}}%
^{(2)}[\rho]==-\int\limits_{\Omega}dm(\mathbf{x})\rho(\mathbf{x})\log
\rho(\mathbf{x}) \label{LL56}%
\end{equation}

\subsubsection{Symmetry properties}

Recall that the R\'{e}nyi entropy is invariant under the group $\Gamma
(\Omega,m)$ of invertible transformations of $\Omega$ that preserve the
measure $m$. A measure preserving transformation $\Omega\overset
{T}{\rightarrow}\Omega$ of the measure space $(\Omega,\mathcal{A},m)$ is a
transformation such that for all $A\subset\Omega$, $A\in\mathcal{A}$, we have
$m\left[  T^{-1}\left(  A\right)  \right]  =m\left(  A\right)  $. Define the
operator $U_{T}$ acting on the distribution function as
\begin{align*}
\rho &  \rightarrow U_{T}\rho=\rho^{\prime}\\
\rho^{\prime}(x)  &  :=\rho\left[  T(x)\right]
\end{align*}
Then it is easy to verify that both Tsallis and R\'{e}nyi entropies are
invariant: for all $\rho$ such that $S_{R,q}[\rho]$ is finite we have for all
$T\in\Gamma(\Omega,m)$
\[
S_{R,q}[\rho]=S_{R,q}[U_{T}\rho]
\]
In the case of the classical definition of the R\'{e}nyi entropy, when the
measure space is discrete and the measure $m$ is the counting measure, the
group $\Gamma(\Omega,m)$ is the group generated by permutations of finite
subsets of elements of $\Omega$. Clearly, from Eq.(\ref{LL8.1}), $S_{R,q}%
[\rho]$ is invariant under permutations. This property was one of the axioms
in the axiomatic definitions of the classical R\'{e}nyi entropy.

\noindent We denote by $\Gamma(\Omega_{j}~,m_{j})~$ the group of measures
preserving transformations of the measure space $(\Omega_{j},\mathcal{A}%
_{j},m_{j})$ from the decomposition of $(\Omega,\mathcal{A},m)$ from
Eqs.(\ref{LL9} -\ref{LL11}). Then we have the following

\begin{proposition}
\label{markProp_GRE_invariance} The GRE $S_{\mathbf{p}}^{(1)}[\rho]$,
$S_{\mathbf{q}}^{(2)}[\rho]$, defined by Eqs.(\ref{LL18}, \ref{LL26}), is
invariant under the sub group $\Gamma(\Omega_{1}~,m_{1})\times\Gamma
(\Omega_{2}~,m_{2})\times...\times\Gamma(\Omega_{N}~,m_{N})$ of the full group
$\Gamma(\Omega,m)$. Let $T_{k}$ be a transformation of the space $\Omega_{k}$
that preserves the measure $m_{k}$. In other words $T_{k}\in\Gamma(\Omega
_{k}~,m_{k})$. Define
\[
\rho^{\prime}(x_{1},...,x_{N}):=\rho\left[  T_{1}(x_{1}),...,T_{N}
(x_{N})\right]
\]
Then we have the invariance properties
\begin{align*}
S_{\mathbf{p}}^{(1)}[\rho]  &  =S_{\mathbf{p}}^{(1)}[\rho^{\prime}]\\
S_{\mathbf{q}}^{(2)}[\rho]  &  =S_{\mathbf{q}}^{(2)}[\rho^{\prime}]
\end{align*}
with $S_{\mathbf{p}}^{(1)}[\rho]$, $S_{\mathbf{q}}^{(2)}[\rho]$ defined in
Eqs.(\ref{LL18}, \ref{LL26}).
\end{proposition}

Fore more details on the symmetry properties of GRE with respect to measure
preserving transformations, see ref. \cite{SGySASGCategory},
\cite{SGYSGstructure}.

\subsubsection{Geometric properties}

Beyond the physical applications, the previous geometric definitions of the
R\'{e}nyi entropy and GRE are more advantageous; in our approach the basic
objects are the norms defined in Eqs.(\ref{LL13}-\ref{LL17}) or pseudo-norms
defined in Eqs.(\ref{LL20} -\ref{LL25}). In the case when $p_{k}\geq1$,
$\mathbf{p}=\{p_{1},...,p_{N}\}$, (see \cite{Besov}) the norm $\left\Vert
.\right\Vert _{\mathbf{p}}$ has the usual properties: for $a\in\mathbf{R}$ we
have $\left\Vert a\rho\right\Vert _{\mathbf{p,}m}=\left\vert a\right\vert
\left\Vert \rho\right\Vert _{\mathbf{p,}m}$, respectively \cite{Besov}
\begin{equation}
\left\Vert \rho_{1}+\rho_{2}\right\Vert _{\mathbf{p,}m}\leq\left\Vert \rho
_{1}\right\Vert _{\mathbf{p,}m}+\left\Vert \rho_{2}\right\Vert _{\mathbf{p,}m}
\label{LL56.1}%
\end{equation}
In particular, it follows the convexity of the mapping $\rho\rightarrow
\left\Vert \rho\right\Vert _{\mathbf{p}}$ : for $0\leq\alpha\leq1$ we have
\begin{equation}
\left\Vert \alpha\rho_{1}+\left(  1-\alpha\right)  \rho_{2}\right\Vert
_{\mathbf{p,m}}\leq\alpha\left\Vert \rho_{1}\right\Vert _{\mathbf{p,m}%
}+\left(  1-\alpha\right)  \left\Vert \rho_{2}\right\Vert _{\mathbf{p,}m}
\label{LL57}%
\end{equation}
In the case $0<q_{k}\leq1$, $\mathbf{q}=\{q_{1},...,q_{N}\}$, the properties
of the pseudo-norms $N[\rho]_{\mathbf{q},m}$ defined in Eqs.(\ref{LL20}
-\ref{LL25}) also allows geometrical interpretations. We have
\begin{equation}
N_{\mathbf{q,m}}[\rho_{1}+\rho_{2}]\leq N[\rho_{1}]_{\mathbf{q,m}}+N[\rho
_{2}]_{\mathbf{q,m}} \label{LL58}%
\end{equation}
This can be proven recursively by using the definition and the simple
inequality $|x+y|^{q}\leq|x|^{q}+|y|^{q}$, with $0<q\leq1$. Instead of
convexity we have the following concavity inequality \textbf{in the first
octant only} ($\rho_{1,2}\geq0$)
\begin{equation}
N_{\mathbf{q,m}}[\alpha\rho_{1}+\left(  1-\alpha\right)  \rho_{2}]\geq\alpha
N_{\mathbf{q,m}}[\rho_{1}]+(1-\alpha)N_{\mathbf{q,m}}[\rho_{2}] \label{LL59}%
\end{equation}
The Eq.(\ref{LL59}) can be proven recursively by using the concavity of the
function $f(x):=x^{q}$ with $0<q\leq1$.

By defining, the distance function between distribution functions $\rho_{1}$
and $\rho_{2}$ in the infinite dimensional space of PDF's by $d(\rho_{1}%
,\rho_{2}):=\left\Vert \rho_{1}-\rho_{2}\right\Vert _{\mathbf{p}}$ for
$p_{k}\geq1$ and $d(\rho_{1},\rho_{2}):=N_{\mathbf{q,m}}[\rho_{1}-\rho_{2}]$
for $0<q_{k}\leq1$, respectively, we have the triangle inequality
\begin{equation}
d(\rho_{1},\rho_{3})\leq d(\rho_{1},\rho_{2})+d(\rho_{2},\rho_{3})
\label{LL59.1}%
\end{equation}
allowing the geometrical interpretation of GRE in term of distance in the
functional space of admissible PDF's. The more general functional
$D_{\mathbf{p},m}(\rho)$ is exposed in Appendix, subsection
\ref{markerSubsectionPropertiesOfFunctionalDpmf}.

\subsection{Applications of GRE.}

\ In the work \cite{SonninoSteinbrGRE} was proven the H-Theorem: \ in the case
$N=2$, the GRE is a Liapunov functional for a class of random dynamical
systems that describe the anomalous transport in plasmas \cite{Balescu1},
\cite{balescu3}. This property can be generalized easily for the $N>2$ case.
Recall that by using the Maximal Entropy principle for the RTE it is possible
to obtain probability distribution functions with algebraic decay, similar to
the derivation of the normal distribution \cite{CLTmaxent}. \ By using the
MaxEnt principle for GRE, a class of PDF with algebraic decay in one variable
(for a particular parameter value it is the symmetric stable Cauchy-Lorentz
distribution), and Gaussian in the second variable was derived
\cite{SonninoSteinbrGRE}. From qualitative point of view such kind of mixed
behavior is typical for the full PDF (the joint PDF of the driving and driven
system) in random linear stochastic processes (see \cite{sgbw}, \cite{sgw}).

By computing the GRE\ of the full PDF of a complex dynamical system, that
contains a Hamiltonian subsystem supposedly driven by an external dynamical
system, it is possible in a systematic way to detect the existence or absence
of the back reaction \cite{SGYSGstructure}.

In the case when $\Omega=\mathbb{R}^{N}$, $dm(x)=d^{N}x$ , the \ $N$
dimensional volume element, it is easy to prove that despite the R\'{e}nyi
entropy is sensitive to the rescaling, the variation of the R\'{e}nyi entropy
\[
S_{R,p,m}[\rho]-S_{R,q,m}[\rho]
\]
is invariant with respect of the full group of affine transformation of the
space $\mathbb{R}^{N}$, including rescaling, so it can \ used as a first
criteria to identify distributions that differs by rescaling and Euclidean
motions. This result can be generalized to GRE. Consider for instance in the
case $N=2$, $\Omega_{k}=\mathbb{R}^{N_{k}}$, $dm_{k}(x_{k})=d^{N_{k}}x_{k}$
with $k=\overline{1,2}$ . \ Let denote
\begin{align*}
N_{q_{1},q_{2}}  &  =\int\limits_{\mathbb{R}^{N_{1}}}d^{N_{1}}\mathbf{x}%
_{1}\left[  \int\limits_{\mathbb{R}^{N_{2}}}d\mathbf{x}_{2}\rho(\mathbf{x}%
_{1},\mathbf{x}_{2})^{q_{2}}\right]  ^{q_{1}}\\
S_{q_{1},q_{2}}  &  =\frac{1}{1-q_{2}}\log N_{q_{1},q_{2}}\\
T_{q_{1},q_{2}}  &  =S_{q_{1},q_{2}}\frac{1-q_{2}}{q_{1}}%
\end{align*}

\bigskip

While $S_{q_{1}q_{2}}$ is not invariant on the general affine group
\begin{align*}
\mathbb{R}^{N_{1}}  &  \ni\mathbf{x}_{1}\rightarrow A_{1}\mathbf{x}%
_{1}+\mathbf{b}_{1}\in\mathbb{R}^{N_{1}}\\
\mathbb{R}^{N_{2}}  &  \ni\mathbf{x}_{2}\rightarrow A_{2}\mathbf{x}%
_{2}+\mathbf{b}_{2}\in\mathbb{R}^{N_{2}}%
\end{align*}
the linear combination of GRE's, \ with $\overline{q_{1}}\neq q_{1}$,
\ $\overline{q_{2}}\neq q_{2}$, defined as follows
\[
T_{\overline{q_{1}},\overline{q_{2}}}-T_{\overline{q_{1}},q_{2}}%
-T_{q_{1},\overline{q_{2}}}+T_{q_{1},q_{2}}%
\]
is an invariant and can be used to the classification of probability distributions.

A strategy to use the R\'{e}nyi distribution is to compute \ for a large range
of parameters. By\ \ \ \ Lemma on Rearrangements from \cite{SGYSGstructure},
in the case of large class of measured spaces, from the equality
\begin{equation}
S_{R,q,m}[\rho_{1}]=S_{R,q,m}[\rho_{2}] \label{aplic1}%
\end{equation}
\textbf{for all }$q\in(a,b)$,results that $\rho_{1},\rho_{2}$ are related by
\begin{equation}
\rho_{1}(x)=\rho_{2}(T(x)) \label{aplic10}%
\end{equation}
where $x\rightarrow T(x)$ is a map (not necessary continuos) that preserves
the measure $m$. This observation allows to identify PDF's that are related by
a measure-preserving coordinate change. \ Nevertheless there exists situations
where the range of distributions $\rho_{2}$ that for given $\rho_{1}$
\ satisfy Eq.(\ref{aplic1}) is too large.\ Let a fixed measure space
$(\Omega,\mathcal{A},m)$ and consider two separate time series $x(t),y(t)\in
\Omega$ , \ that that is either a stochastic process or a deterministic
process with random initial conditions, denote $x_{k}=x(t_{k})$,
$y_{k}=y(t_{k})$. Denote by $\rho^{(x)}(x_{1},x_{2})$ respectively by
$\rho^{(y)}(y_{1},y_{2})$ the corresponding joint PDF. If their R\'{e}nyi
entropies are equal, or, by Eq.(\ref{LL4.3}, \ref{LL5})%
\[
\int\limits_{\mathbb{\Omega}}dx_{1}\int\limits_{\mathbb{\Omega}}dx_{2}%
\rho^{(x)}(x_{1},x_{2})^{q}=\int\limits_{\mathbb{\Omega}}dy_{1}\int
\limits_{\mathbb{\Omega}}dy_{2}\rho^{(y)}(y_{1},y_{2})^{q}%
\]
for all $q\in(a,b)$, according to Lemma on Rearrangements from
\cite{SGYSGstructure} there exists a map $\ \Omega\times\Omega\ni$
$(x_{1},x_{2})\rightarrow(y_{1},y_{2})=T(x_{1},x_{2})\in\Omega\times\Omega$
that preserves the measure $dm(x_{1})dm(x_{2})=dm(y_{1})dm(y_{2})$ such that%
\begin{equation}
\rho^{(y)}(T(x_{1},x_{2}))=\rho^{(x)}(x_{1},x_{2}) \label{aplic20}%
\end{equation}

The degree of indeterminacy from Eq.(\ref{aplic20}) can be reduced by if we
compare the numerical values of the GRE corresponding to $N=2$ case, or by
Eqs.(\ref{LL20}-\ref{LL26})
\begin{align}
\int\limits_{\mathbb{\Omega}}dx_{1}\left[  \int\limits_{\mathbb{\Omega}}%
dx_{2}\rho^{(x)}(x_{1},x_{2})^{q_{2}}\right]  ^{q_{1}}  &  =\int
\limits_{\mathbb{\Omega}}dy_{1}\left[  \int\limits_{\mathbb{\Omega}}dy_{2}%
\rho^{(y)}(y_{1},y_{2})^{q_{2}}\right]  ^{q_{1}}\label{aplic30}\\
\int\limits_{\mathbb{\Omega}}dx_{2}\left[  \int\limits_{\mathbb{\Omega}}%
dx_{1}\rho^{(x)}(x_{1},x_{2})^{q_{2}}\right]  ^{q_{1}}  &  =\int
\limits_{\mathbb{\Omega}}dy_{2}\left[  \int\limits_{\mathbb{\Omega}}dy_{1}%
\rho^{(y)}(y_{1},y_{2})^{q_{2}}\right]  ^{q_{1}} \label{aplic40}%
\end{align}

for $\ a_{1}<q_{1}<b_{1}$, $a_{2}<q_{2}<b_{2}$. According to Lemma on
Rearrangements from \cite{SGYSGstructure} there exists a maps $\ \Omega_{1}%
\ni$ $x_{1}\rightarrow y_{1}=T_{1}(x_{1})\in\Omega_{1}$ and $\Omega_{2}\ni$
$x_{2}\rightarrow y_{2}=T_{2}(x_{2})\in\Omega_{2}$\ \ that preserves the
measures $dm(x_{1})=dm(y_{1})$ , $dm(x_{2})=dm(y_{2})$ such that%
\[
\rho^{(y)}(T_{1}(x_{1}),T_{2}(x_{2})))=\rho^{(x)}(x_{1},x_{2})
\]
In conclusion, the use of GRE provide a finer classification of the PDF,
classification that use measure theoretic aspects, in the case when we use the
GRE for a large set of values of the parameters $q_{1},q_{2}$.

In the case of analytic models of parametric destabilizations \cite{RLASATO}%
,\cite{sgbw}, \cite{sgw}, gyrokinetic simulations of the micro instabilities
in the tokamak plasma \cite{YanickXavierGuillaume} or more generally,
stochastic processes related to self-organized criticality
\cite{BakTangWiesenfeld}-\cite{GarciaCarreras}, one of the problems is that at
least in the theoretical models the mean value $\mathbb{E}(|x(t)|)$ is
infinite or practically it is highly fluctuating. In this case the study of
the long time correlation decay can be performed by studying, for example the
GRE of the joint PDF for $x_{k}=x(t_{k})$ where $k=\overline{1,3}$ with
$t_{2}-t_{1}\rightarrow\infty$ and $t_{3}-t_{2}\rightarrow\infty$. In the case
of stationary regime (see Eq.(\ref{LL50})) we have the asymptotic
factorization
\begin{align*}
N_{\mathbf{q}}[\rho] &  =\int\limits_{\mathbb{R}}dx_{1}\left[  \int
\limits_{\mathbb{R}}dx_{2}\left[  \int\limits_{\mathbb{R}}dx_{3}\rho
(x_{1},x_{2},x_{3})^{q_{3}}\right]  ^{q_{2}}\right]  ^{q_{1}}\rightarrow\\
&  \left[  \int\limits_{\mathbb{R}}dx_{2}\left[  \rho_{2}(x_{2})\right]
^{q_{2}q_{3}}\right]  ^{q_{1}}\left[  \int\limits_{\mathbb{R}}dx_{3}\rho
_{3}(x_{3})^{q_{3}}\right]  ^{q_{1}q_{2}}\int\limits_{\mathbb{R}}dx_{1}\left[
\rho_{1}(x_{1})\right]  ^{q_{1}q_{2}q_{3}}%
\end{align*}
where
\[
\rho_{1}(x)=\int\limits_{\mathbb{R}}dx_{2}\int\limits_{\mathbb{R}}dx_{3}%
\rho(x,x_{2},x_{3}),...
\]
or in term of entropies we have the decomposition at large time lags%
\begin{align}
(q_{3}-1)S_{q_{1},q_{2},q_{3}}^{(2)}[\rho] &  \rightarrow q_{1}q_{2}%
(q_{3}-1)S_{R,q_{3}}[\rho_{3}]+\label{aplic50}\\
&  q_{1}(q_{2}q_{3}-1)S_{R,q_{2}q_{3}}[\rho_{2}]+(q_{1}q_{2}q_{3}%
-1)S_{R,q_{1}q_{2}q_{3}}[\rho_{1}]\nonumber
\end{align}
The speed of convergence in Eq.(\ref{aplic50}) can be used to characterize the
correlation decay in stochastic processes, where the usual mean values diverge.

We remark also that in the cases when limiting values for $q_{k}\searrow0$
and/or $q_{j}\rightarrow\infty$ , the GRE\ give information on the support and
extreme value properties of the PDF.

\section{Bound and stability properties of
GRE.\label{markerSectionGREboundStabiltiy}}

By the results and notations from subsection
\ref{markerSubsectLogCcnvManyVriables}, the problem of boundednes of GRE can
be treated directly. In the following we consider, according to the notations
from subsection \ref{markerSubsectGREdeffinition}, the measure space
$(\Omega,\mathcal{A},m)$ Eq.(\ref{LL10}, \ref{LL11}) and probability measure
$dP(\mathbf{x})=\rho dm$ from Eq.(\ref{LL12}). Consider the set of exponents
$(q_{1},...,q_{N}):=\mathbf{q}$ the associated functional $N_{\mathbf{q},m}$
from Eq.(\ref{LL25}) and the hyper rectangle $D_{N}\subset\mathbb{R}^{N}$
defined in Eq.(\ref{logconv180}), Appendix
\ref{markerSubsectLogCcnvManyVriables} and suppose, for technical reasons,
that we are in the generic case: $D_{N}$ has non zero volume. Suppose in the
continuation that at least one of the vertices of the hyper rectangle $D_{N}$
contains the point $q_{i}=1;~1\leq i\leq N$, and denote this point by
$\mathbf{u}$. Due to the normalization condition PDF we have $N_{\mathbf{u}%
,m}(\rho)=1$, so it is plausible to suppose that the functional $N_{\mathbf{q}%
,m}$ is defined also in some neighborhood of $\mathbf{u}=(1,...1)$. According
to Theorem \ref{markerTheoremAppLogConvNvar} the function $\mathbf{q}%
\rightarrow N_{\mathbf{q},m}(\rho)$ is log-convex in the variable $\mathbf{q}%
$. Suppose that on the vertices of the hyper rectangle $D_{N}$ we have the
bounds Eqs.(\ref{logconv190}), where in our case
\begin{equation}
g(w_{1},...w_{N})=N_{\mathbf{w},m}(\rho) \label{LL60}%
\end{equation}
Then, according to the Corollary \ref{markerCorolaryLogconvBoundNvar}, we have
the bounds Eq.(\ref{logconv220.1}) (with the notations in subsection
\ref{markerSubsectLogCcnvManyVriables}, Eqs.(\ref{logconv200}-\ref{1300},
\ref{1350})
\begin{equation}
\log N_{\mathbf{w},m}(\rho)\leq b_{N}(\mathbf{w});\ \mathbf{w}\in D_{N}
\label{LL61}%
\end{equation}
\newline Denote by $V_{N}$ the set of vertices of the hyper rectangle $D_{N}$
and by $V_{N}^{\prime}$ the set of vertices excepting the vertex
$\mathbf{u}=(1,...1)$. Consider \ the exact PDF $\rho(\mathbf{x})$ and the
approximating sequence $\rho_{n}(\mathbf{x})$
\begin{equation}
\int\limits_{\Omega}dm(\mathbf{x})\left\vert \rho_{n}(\mathbf{x}%
)-\rho(\mathbf{x})\right\vert dm(\mathbf{x})\leq\varepsilon_{n}\underset
{n\rightarrow\infty}{\rightarrow}0 \label{LL62}%
\end{equation}
or, equivalently, (Eqs. \ref{LL20}-\ref{LL25})
\begin{equation}
N_{\mathbf{u},m}(\rho_{n}-\rho)\leq\varepsilon_{n}\underset{n\rightarrow
\infty}{\rightarrow}0 \label{LL63}%
\end{equation}
Suppose that on the rest of the vertices $\mathbf{v}\in V_{N}^{\prime}$ we
have the bounds
\begin{align}
N_{\mathbf{v},m}(\rho_{n})  &  \leq B;~\mathbf{v}\in V_{N}^{\prime
}\label{LL64}\\
N_{\mathbf{v},m}(\rho)  &  \leq B;;~\mathbf{v}\in V_{N}^{\prime} \label{LL65}%
\end{align}
where $B$ is a constant. Denote by $D_{N}^{\prime}$ the subset of the hyper
rectangle the set $Int\ (D_{N})$, the set of interior points of $D_{N}$,
defined by $Int(D_{N}):=\{\mathbf{w}|a_{k}^{(1)}<w_{k}<a_{k}^{(2)}%
~;~k=\overline{1,N}\}$. By our previous technical assumption, $D_{N}^{\prime}$
is non void. The following stability result will be proved:

\begin{proposition}
\label{markerPropositionGREstability}Under previous conditions Eqs.(\ref{LL62}
-\ref{LL65}) for all $\mathbf{w\in}D_{N}^{\prime}$ we have \
\begin{align}
&  N_{\mathbf{w},m}(\rho_{n})\underset{n\rightarrow\infty}{\rightarrow
}N_{\mathbf{w},m}(\rho)\label{LL67}\\
&  S_{\mathbf{q}}^{(2)}[\rho_{n},m]\underset{n\rightarrow\infty}{\rightarrow
}S_{\mathbf{q}}^{(2)}[\rho,m] \label{LL68}%
\end{align}

\end{proposition}

\bigskip

\begin{proof}
The bounds from Eqs.(\ref{LL64}, \ref{LL65}) can be translated in terms of
distances $D_{\mathbf{p},m}\left[  \rho\right]  $, by using Eqs.(\ref{33.9}
-\ref{33.16}) :
\begin{align}
D_{\mathbf{v},m}(\rho_{n})  &  \leq B^{\prime};~\mathbf{v}\in V_{N}^{\prime
}\label{LL69}\\
D_{\mathbf{v},m}(\rho)  &  \leq B^{\prime};~\mathbf{v}\in V_{N}^{\prime
}\label{LL70}\\
B^{\prime}  &  =\underset{\mathbf{v}\in V_{N}^{\prime}}{\max}\left(
B^{i(v_{1})}\right) \nonumber
\end{align}
From Eqs.(\ref{33.7}, \ref{LL69}, \ref{LL70}) results
\begin{equation}
D_{\mathbf{v},m}(\rho_{n}-\rho)\leq2B^{\prime};~\mathbf{v}\in V_{N}^{\prime}
\label{LL71}%
\end{equation}
This set of bounds we rewrite again in the term of $N_{\mathbf{w},m}$, by
using Eqs.(\ref{33.9}-\ref{33.16}):
\begin{align}
N_{\mathbf{v},m}(\rho_{n}-\rho)  &  \leq B";~\mathbf{v}\in V_{N}^{\prime
}\label{LL72}\\
B"  &  =\underset{\mathbf{v}\in V_{N}^{\prime}}{\max}\left(  2B^{\prime
}\right)  ^{1/i(v_{1})}\nonumber
\end{align}
Now we have bounds on all of the vertices of $D_{N}$ and we can use Corrolary
\ref{markerCorolaryLogconvBoundNvar} and Eq.(\ref{1350}) with $g(\mathbf{w}%
)=N_{\mathbf{w},m}(\rho_{n}-\rho)$. We get:
\begin{align}
\log N_{\mathbf{w},m}(\rho_{n}-\rho)  &  \leq b_{N}(\mathbf{w});~\mathbf{w}\in
D_{N}^{\prime}=Int(D_{N})\label{LL73}\\
b_{N}(\mathbf{w})  &  :=\sum\limits_{\mathbf{v}\in V_{N}}P(\mathbf{v}%
,\mathbf{w})\log A_{\mathbf{v}}^{\prime}\label{LL74}\\
0  &  <P(\mathbf{v},\mathbf{w})<1;~\mathbf{w}\in D_{N}^{\prime} \label{LL75}%
\end{align}
Because $V_{N}=V_{N}^{\prime}\cup\{\mathbf{u}\}$ \ the Eq.(\ref{LL74}) can be
rewritten as follows
\begin{align}
b_{N}(\mathbf{w})  &  =\sum\limits_{\mathbf{v}\in V_{N}}P_{\mathbf{v}%
}(\mathbf{w})\log~A_{\mathbf{v}}\label{LL77}\\
0  &  <P_{\mathbf{v}}(\mathbf{w})<1;~\mathbf{w}\in D_{N}^{\prime} \label{LL78}%
\end{align}
and $A_{\mathbf{v}}$ are the bounds on $g(\mathbf{w})$ on the vertices $V_{N}%
$. By using Eqs.(\ref{LL63}, \ref{LL72}) we rewrite Eq.(\ref{LL77}) as
follows
\begin{align}
\log N_{\mathbf{w},m}(\rho_{n}-\rho)  &  \leq b_{N}(\mathbf{w})=P_{\mathbf{u}%
}(\mathbf{w})\log\varepsilon_{n}+K(\mathbf{w});\mathbf{w}\in D_{N}^{\prime
}\label{LL79}\\
K(\mathbf{w})  &  =\sum\limits_{\mathbf{v}\in V_{N}^{\prime}}P_{\mathbf{v}%
}(\mathbf{w})\log~B^{\prime\prime} \label{LL80}%
\end{align}
From Eqs.(\ref{LL62}, \ref{LL78}-\ref{LL80}) results
\[
N_{\mathbf{w},m}(\rho_{n}-\rho)\underset{n\rightarrow\infty}{\rightarrow
}0;\ \mathbf{w}\in D_{N}^{\prime}%
\]
By using Eqs.(\ref{33.9}, \ref{33.16})
\[
D_{\mathbf{w},m}(\rho_{n}-\rho)\underset{n\rightarrow\infty}{\rightarrow
}0;\ \mathbf{w}\in D_{N}^{\prime}%
\]
From Eq.(\ref{33.8}) we obtain
\[
\left\vert D_{\mathbf{w},m}(\rho_{n})-D_{\mathbf{w},m}(\rho)\right\vert \leq
D_{\mathbf{w},m}(\rho_{n}-\rho)\underset{n\rightarrow\infty}{\rightarrow
}0;\ \mathbf{w}\in D_{N}^{\prime}%
\]
and using again Eqs.(\ref{33.9}, \ref{33.16}) we find the requested stability
results
\begin{align*}
D_{\mathbf{w},m}(\rho_{n})\underset{n\rightarrow\infty}{\rightarrow
}D_{\mathbf{w},m}(\rho);\ \mathbf{w}  &  \in D_{N}^{\prime}\\
N_{\mathbf{w},m}(\rho_{n})\underset{n\rightarrow\infty}{\rightarrow
}N_{\mathbf{w},m}(\rho);\ \mathbf{w}  &  \in D_{N}^{\prime}%
\end{align*}
and from Eq.(\ref{33.17}) and normalization of the PDF's $\rho,\rho_{n}$
\[
S_{\mathbf{w}}^{(2)}[\rho_{n},m]\underset{n\rightarrow\infty}{\rightarrow
}S_{\mathbf{w}}^{(2)}[\rho_{n},m];\ \mathbf{w}\in D_{N}^{\prime}%
\]
which completes the proof.
\end{proof}

\begin{remark}
Note that it is possible the extend the proof to the case when the volume of
$D_{N}$ is zero, or to extend the stability proof to the part of boundary of
$D_{N}$, where $P_{\mathbf{u}}(\mathbf{w})>0$.
\end{remark}

\section{Conclusions}

We proved that, in the general case of measure spaces, the entropies defined
by C. Tsallis (TE), A. R\'{e}nyi (RE) as well as the generalized R\'{e}nyi's
entropy (GRE), are well defined concepts and can be computed in a numerically
stable manner for a large range of parameters, if a stabilizing condition is
imposed. However, for the case of Shannon-Boltzmann's entropy (BSE) two
stabilizing conditions are necessary. In all cases the stabilizing conditions
are expressed in the term of finiteness of Lebesgue space $L^{p}$ norm of the
probability density functions. As a mathematical by-product, we proved the
logarithmic convexity of the integrals related to generalized Lebesgue space norms.

\begin{acknowledgement}
The authors are grateful to Prof. M. Van Schoor and Dr D. Van Eester from
Royal Military School, Brussels. Gy. Steinbrecher is grateful to Prof. S.
Ciulli, University Montpelier, G. Nenciu, Institute of Mathematics "Simion
Stoilow" of the Romanian Academy, I. Sabba \c{S}tef\~{a}nescu, University
Karlsruhe and C. Pomponiu for useful discussions. Giorgio Sonnino is also
grateful to Prof. P. Nardone and Dr. P. Peeters of the Universit\'{e} Libre de
Bruxelles (ULB) for useful discussions and suggestions.
\end{acknowledgement}

\section{Appendix \label{markerSectionAppndx}}

\subsection{Analytic extrapolation
theorem.\label{markerSubsectExtrapolTheorem}}

For easy reference, we shall prove in a special case the following theorem
\cite{CiulliStabProblems}, \cite{CiulliNenciu}, \cite{CiulliPhysRep}. Denote
by $D$ the domain in the complex plain $\mathbb{C}$ defined as follows
\begin{equation}
D=\left\{  z|z\in\mathbb{C},~0<\operatorname{Im}(z)<b,~a_{1}<\operatorname{Re}%
(z)<a_{2}\right\}  \label{ae1}%
\end{equation}
The boundary of the domain is $\partial D=\Gamma_{0}\cup\Gamma_{2}\cup
\Gamma_{2}\cup\Gamma_{3}$ \ where
\begin{align*}
\Gamma_{0}  &  =\left\{  z|z\in\mathbb{C},~\operatorname{Im}(z)=0,~a_{1}%
<\operatorname{Re}(z)<a_{2}\right\} \\
\Gamma_{1}  &  =\left\{  z|z\in\mathbb{C},b\geq~\operatorname{Im}%
(z)\geq0,~\operatorname{Re}(z)=a_{1}\right\} \\
\Gamma_{2}  &  =\left\{  z|z\in\mathbb{C},~b\geq~\operatorname{Im}%
(z)\geq0,~\operatorname{Re}(z)=a_{2}\right\} \\
\Gamma_{3}  &  =\left\{  z|z\in\mathbb{C},~\operatorname{Im}(z)=b,~a_{1}%
<\operatorname{Re}(z)<a_{2}\right\}
\end{align*}
Denote by $\mathcal{K}_{\varepsilon}$ the family of analytic functions in $D$,
such that for all $f(z)\in\mathcal{K}_{\varepsilon}$ we have
\begin{align}
|f(z)|  &  \leq\varepsilon;~z\in\Gamma_{0}\label{ae2}\\
|f(z)|  &  \leq m;~z\in\Gamma_{1}\cup\Gamma_{2}\cup\Gamma_{3} \label{ae3}%
\end{align}
We shall prove the following theorem

\begin{theorem}
\label{markerTheoremExtrapol}Let $z_{0}\in D$ \ and the constant $m$ fixed.
Under the previous conditions Eqs.(\ref{ae2}, \ref{ae3}), for fixed $m$ and
for all $f(z)\in\mathcal{K}_{\varepsilon}$ we have
\begin{equation}
|f(z_{0})|\leq\delta(\varepsilon) \label{ae5}%
\end{equation}
where $\delta(\varepsilon)$ does not depend on $f$ and
\begin{equation}
\underset{\varepsilon\rightarrow0}{\lim}\delta(\varepsilon)=0 \label{ae6}%
\end{equation}

\end{theorem}

\begin{proof}
Denote by $\Gamma_{m}:=\Gamma_{1}\cup\Gamma_{2}\cup\Gamma_{3}$ and by
$u_{0}(z)$, $u_{1}(z)$ the harmonic functions in $D$ with the following
Dirichlet boundary conditions.
\begin{align}
u_{0}(z)  &  =1;~z\in\Gamma_{0}\label{ae6.1}\\
u_{0}(z)  &  =0;~z\in\Gamma_{m}\label{ae6.2}\\
u_{1}(z)  &  =1;~z\in\Gamma_{m}\label{ae6.3}\\
u_{1}(z)  &  =0;~z\in\Gamma_{0} \label{ae6.4}%
\end{align}
According to the standard terminology \cite{Nevanlinna}, \cite{Oksendal},
\cite{NenciuHamonicMeasure} the function $u_{0}(z)~$\ is the harmonic measure
of the subset $\Gamma_{0}$ while $u_{1}(z)$ is the harmonic measure of
$\Gamma_{m}.$ We remark that
\begin{equation}
0<u_{k}(z_{0})<1;\ k=0,1 \label{ae7}%
\end{equation}
with strict inequalities, because $z_{0}$ is an interior point. This property
is a consequence of the interpretation of $u_{k}(z_{0})$ as hitting
probability of planar Brownian motion \cite{Oksendal}. These inequalities has
also obvious physical meaning in the case when $u_{k}(z)$ are stationary
temperature fields. Denote
\begin{equation}
U_{\varepsilon}(z):=u_{0}(z)\log\varepsilon+u_{1}(z)\log m \label{ae8}%
\end{equation}
Also denote by $V_{\varepsilon}(z)$ the real harmonic conjugate of
$U_{\varepsilon}(z)$, such that the function%
\begin{equation}
K_{\varepsilon}(z):=U_{\varepsilon}(z)+iV_{\varepsilon}(z) \label{ae9}%
\end{equation}
is an analytic function in $D$. Denote
\begin{equation}
C_{\varepsilon}(z):=\exp K_{\varepsilon}(z) \label{ae10}%
\end{equation}
From Eqs.(\ref{ae6.1}-\ref{ae10}) we have that $C_{\varepsilon}(z)$ is
analytic, without zeroes in $D$ and have the following properties
\begin{align}
|C_{\varepsilon}(z)|  &  =\varepsilon;~z\in\Gamma_{0}\label{ae11}\\
|C_{\varepsilon}(z)|  &  =m;~z\in\Gamma_{m}\label{ae12}\\
\underset{^{e\rightarrow0}}{\lim}C_{\varepsilon}(z_{0})  &  =0;~~z_{0}\in D
\label{ae14}%
\end{align}
The function
\begin{equation}
G(z):=f(z)/C_{\varepsilon}(z) \label{ae14.1}%
\end{equation}
is analytic in $D$ and from Eqs.(\ref{ae2}, \ref{ae3}, \ref{ae11}, \ref{ae12})
we have that $|G(z)|\leq1$ for all $z\in\partial D$. From the maximum modulus
principle \cite{Rudin} we get
\begin{equation}
|G(z)|\leq1;~\forall z\in\partial D\cup D \label{ae15}%
\end{equation}
By denoting $\delta(\varepsilon):=C_{\varepsilon}(z_{0})$, from
Eqs.(\ref{ae14.1}, \ref{ae15}) we have that $|f(z_{0}|\leq$ $\delta
(\varepsilon)$, which according to Eq.(\ref{ae14}) completes the proof.
\end{proof}

\bigskip

\subsection{Some convergence results \label{markerSubsectCalculSumeIntegrale}}

\subsubsection{Discrete case, counterexample 1, subsection
\ref{markerSubsectCounterxDiscrete} \label{markerAppendCounterexDiscrete}}

\paragraph{ Proof of the inequalities (\ref{LL8.8}, \ref{LL8.11})}

It is sufficiently to prove Eq.(\ref{LL8.11}), the convergence in
Eq.(\ref{LL8.8}) results by comparison: $K_{n}<K$. Because in Eq.(\ref{LL8.11}%
) the sequence of $p_{k}$ is monotone decreasing, we use the integral
criteria: we have to prove that
\begin{equation}
K<\underset{N\rightarrow\infty}{\lim}\int\limits_{0}^{N}dk\frac{1}%
{(k+4)\left[  \log\left(  k+4\right)  \right]  ^{2}}<M \label{acex1}%
\end{equation}
for some $M>0$. By change of variable
\begin{equation}
k=\exp(x)-4 \label{chvar}%
\end{equation}
the Eq.(\ref{acex1}) is reduced to the obvious bound for $N$ large
\[
\int\limits_{\log(4)}^{\log(N+4)}\frac{dx}{x^{2}}<M
\]
In conclusion $M=1/\log(4)$ and
\begin{equation}
K_{n}<K<M \label{acex1.1}%
\end{equation}
By considering only the first term in the infinite sum, we also obtain from
Eq.(\ref{LL8.8}),
\begin{equation}
\frac{1}{5[\log(5)]^{2+1/n}}<K_{n} \label{acex1.2}%
\end{equation}

\paragraph{Proof of Eq.(\ref{LL8.12})}

We use Eqs.(\ref{LL8.4}, \ref{LL8.7}, \ref{LL8.8})
\begin{align}
S_{cl}[\mathbf{p}^{(n)}]  &  =A^{(n)}+B^{(n)}+C^{(n)}\label{cex2}\\
A^{(n)}  &  :=\sum\limits_{k=1}^{\infty}p_{k}^{(n)}\log K_{n}\label{cex3}\\
B^{(n)}  &  :=(2+\frac{1}{n})\sum\limits_{k=1}^{\infty}p_{k}^{(n)}\log
\log(k+4)\label{cex4}\\
C^{(n)}  &  :=\sum\limits_{k=1}^{\infty}p_{k}^{(n)}\log(k+4) \label{cex5}%
\end{align}
From Eq.(\ref{LL8.5}) we have $A^{(n)}=\log K_{n}$ and from Eq.(\ref{acex1.1},
\ref{acex1.2}) we get
\begin{equation}
-\log[5(\log(5))^{2+1/n}]<A^{(n)}\leq\log M \label{cex6}%
\end{equation}
We will denote by $C_{1,2,..}$ some fixed constants, that do not depend on
$n$. From Eqs.(\ref{acex1.2}, \ref{cex4}) we have
\begin{equation}
B^{(n)}<C_{1}\sum\limits_{k=1}^{\infty}\frac{1}{(k+4)\left[  \log\left(
k+4\right)  \right]  ^{2}}\log\log(k+4) \label{cex6.1}%
\end{equation}
Let $N$ such that for all $k\geq N-1$ the summand in Eq.(\ref{cex6.1}) is
monotone decreasing. We obtain
\begin{equation}
B^{(n)}<C_{2}+C_{1}\sum\limits_{k=N}^{\infty}\frac{1}{(k+4)\left[  \log\left(
k+4\right)  \right]  ^{2}}\log\log(k+4) \label{cx6.2}%
\end{equation}
In order to prove that the infinite sum in Eq.(\ref{cx6.2}) is convergent we
use the integral criteria and the change of variable Eq.(\ref{chvar})
\begin{align}
B^{(n)}  &  <C_{2}+C_{1}\int\limits_{N-1}^{\infty}dk\frac{1}{(k+4)\left[
\log\left(  k+4\right)  \right]  ^{2}}\log\log(k+4)<\label{cex7}\\
&  =C_{2}+C_{1}\int\limits_{\log(4)}^{\infty}dx\frac{\log(x)}{x^{2}}
<+\infty\label{cex8}%
\end{align}
In conclusion we find that $B^{(n)}$ is uniformly bounded. The summand in
Eq.(\ref{cex5}) is monotone, so by the integral criteria we find
\begin{equation}
\int\limits_{0}^{\infty}dk\frac{1}{(k+4)\left[  \log\left(  k+4\right)
\right]  ^{1+1/n}}<C^{(n)}<\int\limits_{1}^{\infty}dk\frac{1}{(k+4)\left[
\log\left(  k+4\right)  \right]  ^{1+1/n}} \label{cex9}%
\end{equation}
We use again Eq.(\ref{chvar}) and we get
\[
n\log(4)<C^{(n)}<n\log(5)
\]
that together to Eqs.(\ref{cex2}, \ref{cex6}, \ref{cex8}) prove
Eq.(\ref{LL8.12})

\subsubsection{Proof of the results from Counter Example 2, subsection
\ref{markerSubsectCounterxContinous}%
\ \ \label{markerAppendixCounterexContinous}}

In the following we shall estimate, or explicitly compute, the integrals by
using the substitution
\begin{equation}
x=\exp(-t) \label{cex10}%
\end{equation}
The normalization constants $M_{n}$, $M$ from Eqs(\ref{cex2.1}-\ref{cex2.4})
can be computed exactly by using Eq.(\ref{cex10})
\begin{align}
M_{n}^{-1}  &  =\frac{1}{\alpha-1}\left[  \left(  \log~2\right)  ^{1-\alpha
}-\left(  \log~n\right)  ^{1-\alpha}\right]  ;~n>2\label{cex11}\\
M^{-1}  &  =\frac{1}{\alpha-1}\left(  \log~2\right)  ^{1-\alpha} \label{cex12}%
\end{align}
It follows that $M_{n}\rightarrow M$, \ $\rho_{n},\rho\in L^{1}$, $\left\Vert
\rho_{n}\right\Vert _{L^{1}}=\left\Vert \rho\right\Vert _{L^{1}}=1$ \ and by
simple calculations and Eq.(\ref{cex10}) results
\[
\left\Vert \rho_{n}-\rho\right\Vert _{L^{1}}=\int\limits_{0}^{1/n}\frac
{Mdx}{x\left(  \log\frac{1}{x}\right)  ^{\alpha}}+\int\limits_{1/n}^{1/2}%
\frac{|M_{n}-M|dx}{x\left(  \log\frac{1}{x}\right)  ^{\alpha}}\rightarrow0
\]
In order to prove Eq.(\ref{cex2.5}) it is sufficient to prove that for
$0<p<1$
\[
\int\limits_{0}^{1/2}\left[  \frac{1}{x\left(  \log\frac{1}{x}\right)
^{\alpha}}\right]  ^{p}dx<\infty
\]
resulting from the inequality $\log~1/x~>\log\ 1/2$, for $0<x<1/2$, or by
direct calculation by using Eq.(\ref{cex10}).

In order to prove Eq.(\ref{cex2.8}) we denote%
\begin{equation}
f(x):=\frac{1}{x\left(  \log\frac{1}{x}\right)  ^{\alpha}} \label{cex13}%
\end{equation}
and from Eqs.(\ref{LL2},~\ref{cex2.1}, \ref{cex2.2}, \ref{cex13}) we obtain
\begin{align}
S_{cl}(\rho_{n})  &  =-\log M_{n}-M_{n}I_{n}\label{cex14}\\
I_{n}  &  =\int\limits_{1/n}^{1/2}f(x)\log f(x)~dx
\end{align}
We use Eqs.(\ref{cex10}, \ref{cex13}) so the last integral is rewritten as
follows
\begin{equation}
I_{n}=\int\limits_{\log2}^{\log n}\frac{dt}{t^{\alpha}}\left(  t-\alpha\log
t\right)  \label{cex16}%
\end{equation}
In the range $1<\alpha<2$ the integral \ $\int\limits_{\log2}^{\infty}%
\frac{dt}{t^{\alpha}}\log t$ is convergent, so the leading term for
$n\rightarrow\infty$ is given by
\[
I_{n}\asymp\frac{1}{2-\alpha}\left[  \log n\right]  ^{2-\alpha}%
\]
that proves Eq.(\ref{cex2.8})

\subsection{ Class of metric vector spaces and the functional $D_{\mathbf{p}
,m}[f]$, associated to GRE \label{markerSubsectionPropertiesOfFunctionalDpmf}}

We expose here a formalism in order to have a clear control on the
pseudo-norms in the case of arbitrary number of variables. The notation from
Eq.(\ref{LL4.1}) will be used in continuation. In analogy to the compact
notation for the pseudo-norm defined in the case of single variables
Eq.(\ref{LL4.2}), we present here another, equivalent, definition with
Eqs.(\ref{33.1}-\ref{33.5}) of the functional $D_{\mathbf{p},m}[f]$. We have
in mind the measure space $(\Omega,\mathcal{A},m)$ and a function
$f(x_{1},x_{2},...,x_{N})$ similar to PDF $\rho(x_{1},x_{2},...,x_{N})$ with
the structure specified in Eqs~(\ref{LL9}-\ref{LL12}). We use the following

\begin{definition}
\label{markerDefDistance}Let $E$ a Metric Vector Space (MVS) with a
real-valued metric (distance to origin) $\delta:E\rightarrow\mathbb{R}_{+}$
such that for $x,y\in E$ the distance is $d(x,y):=\delta(x-y)$\thinspace. The
function $E\backepsilon v\rightarrow\delta(v)\in\mathbb{R}_{+}$ is called
pseudo-norm, if it satisfy the triangle inequality and it is homogenous with
degree $s$
\begin{align}
\delta(u+v)  &  \leq\delta(u)+\delta(v)\label{d00}\\
\delta(\alpha v)  &  =\left\vert \alpha\right\vert ^{s}\delta(v);\ \label{d02}%
\\
0  &  <s\leq1\\
\delta(v)  &  =0\Rightarrow v=0
\end{align}
This metric vector space $E$ with pseudo-norm $\delta$ will be denoted
$(E,\delta)$.
\end{definition}

Recall that in the case of norms in general we have $s=1$ and in the case of
pseudo norms $N\,_{\mathbf{p,m}}(\rho)$ from Eq.(\ref{LL25}) $s=\prod
\limits_{k=1}^{N}q_{k}<1$. In the case of distance defined previously in
Eqs.(\ref{LL4.2}, \ref{33.5}) $s=\prod\limits_{k=1}^{N}ki(p_{k})<1$ (see
below). Note that $s\leq1$ results from Eqs.(\ref{d00}, \ref{d02}). If $s<0$,
then \thinspace$\delta(\frac{1}{n}v)=n^{-s}\,\delta(v)\rightarrow\infty$ for
$n\rightarrow\infty$, and if $s=0$ then $\delta(\frac{1}{n}v)=\ \,\delta(v)$.
So in both cases the distance is discontinuous near $v=0$.

\begin{remark}
\label{markerRemarkTriangleDiffIneqTopology} From Eq.(\ref{d00}) follows the
useful inequality
\begin{equation}
\left\vert \delta(u)-\delta(v)\right\vert \leq\delta(u-v) \label{d03}%
\end{equation}
We can define the convergence of a sequence of vectors $u_{n}$ to $u$ in $E$
by the distance: $\delta(u_{n}-u)\rightarrow0$. Then Eq.(\ref{d03}) means that
the function $u\rightarrow\delta(u)$ is continuous: $u_{n}\rightarrow
u\Rightarrow\delta(u_{n})\rightarrow d(u)$.
\end{remark}

\subsubsection{ The\ space $L^{p}(\Omega,m,E,d)$ and the\ functional
$D_{p,m,E}(f)$.}

In order to obtain an equivalent definition of $D_{\mathbf{p},\mathbf{m}
}\left[  f\right]  $, from Eq.(\ref{33.5}) we extend the definition of
Eq.(\ref{LL4.2}) in such a manner that the properties Eq.~(\ref{LL4.1}%
-\ref{LL4.5}) of the pseudo-norm are preserved. We have the following

\begin{definition}
\label{markerNotationDistanceLp} Denote \ by $L^{p}(\Omega,m,F,\delta)$ the
vector space of all functions defined on the measure space $(\Omega
,\mathcal{A},m)$ with values in the MVS $(F,\delta)$
\begin{equation}
f:\Omega\rightarrow F \label{d3}%
\end{equation}
such that
\begin{equation}
D_{p,m,F}(f):=\left[  {\int\limits_{\Omega}}\left[  \delta\left[  f\left(
x\right)  \right]  \right]  ^{p}dm(x)\right]  ^{i(p)}<\infty\label{d3.1}%
\end{equation}
where $\delta$ is a pseudo-norm with homogeneity degree $s$.
\end{definition}

We observe that in Eq.(\ref{d3.1}) can be obtained from Eq.(\ref{LL4.2}) by
performing the change
\[
|f|\rightarrow\delta(f)
\]
\ 

\ \ \ \ From Eqs.(\ref{d00}, \ref{d02}\ \ref{d3.1}) we have the following

\begin{proposition}
\label{markerRemarkDistFunctionTriangleHomogeneity} The functional
$D_{p,m,F}(f)$\ from Eq.(\ref{d3.1}) \ is a pseudo-norm (see definition
\ref{markerDefDistance}) and it is homogenous with degree $\sigma=spi(p)$
where $s$ is the homogeneity degree of the pseudo-norm $\delta(.)$.
\end{proposition}

\begin{proof}
The homogeneity results directly from Eqs.(\ref{d02}, \ref{d3.1}). From
Eqs.(\ref{d00}, \ref{d3.1}) we get
\begin{equation}
D_{p,m,F}(f+g)\leq\left[  {\int\limits_{\Omega}}\left[  \delta\left[
f\right]  +\delta\left[  g\right]  \right]  ^{p}dm(x)\right]  ^{i(p)}
\label{dd3.2}%
\end{equation}
In the case $0<p<1$, we have $i(p)=1$ and by using the inequality
$\ (|a|+|b|)^{p}\leq|a|^{p}+|b|^{p}$ we obtain
\begin{equation}
D_{p,m,F}(f+g)\leq D_{p,m,F}(f)+D_{p,m,F}(g) \label{d3.3}%
\end{equation}
In the case $p\geq1$ Eq.(\ref{d3.3}) results from Minkowski inequality for
$L^{p}$ space norms
\end{proof}

Note that according to Remark \ref{markerRemarkTriangleDiffIneqTopology}, we
have
\begin{equation}
\left\vert D_{p,m,F}(f)-D_{p,m,F}(g)\right\vert \leq D_{p,m,F}(f-g)
\label{d3.4}%
\end{equation}

\subsubsection{Proof of the Proposition \ref{markerPropositionD[f]Pseudonorm}}

The proof is by backward induction, in close analogy to the recurrent
definition Eq.(\ref{33.1}-\ref{33.5}). At each step we define, according to
Definition \ref{markerNotationDistanceLp}, a MVS $(E_{k},\Delta_{k})$. We
shall denote by $\sigma_{k}$ the homogeneity degree of the pseudo-norm
$\Delta_{k}$. In the first step we set in Definition
\ref{markerNotationDistanceLp} \ $\Omega=\Omega_{N}$, $m=m_{N}$, $p=p_{N}$,
\ $F=\mathbb{R}$, \ $\delta(.)=|.|$. Define the corresponding MVS
$(E_{N},\Delta_{N})$ as follows
\begin{align}
E_{N}  &  =L^{p_{N}}(\Omega_{N},m_{N},\mathbb{R},|\ \ |)=L^{p_{N}}(\Omega
_{N},m_{N})\label{d3.5}\\
\Delta_{N}(f_{N})  &  =\left[  \int\limits_{\Omega_{N}}\left\vert
f_{N}\right\vert ^{p_{N}}dm_{N}\right]  ^{i(p_{N})};~f_{N}\in E_{N}
\label{d3.6}%
\end{align}
Note that, at this stage, $(E_{N},\Delta_{N})$ is the standard $L^{p}$ MVS. We
get the same conclusion by the Proposition
\ref{markerRemarkDistFunctionTriangleHomogeneity}: $\Delta_{N}$ is a
pseudo-norm (for $p_{N}\geq1$ it is the standard $L^{p}$ norm), homogenous
with degree $\sigma_{N}=p_{N}i(p_{N})$.

\noindent In the second step, for the sake of clarity, we repeat the previous
construction: we set in Definition \ref{markerNotationDistanceLp}
\ $\Omega=\Omega_{N-1}$, $m=m_{N-1}$, $p=p_{N-1}$, but we select, $F=E_{N}$,
\ $\delta(.)=\Delta_{N}(.)$ from Eqs.(\ref{d3.5}, \ref{d3.6}). Define the
corresponding MVS $(E_{N-1},\Delta_{N-1})$ as follows.
\begin{align}
E_{N-1}  &  =L^{p_{N-1}}(\Omega_{N-1},m_{N-1},E_{N},\Delta_{N})\label{d3.7}\\
\Delta_{N-1}(f_{N-1})  &  =\left[  \int\limits_{\Omega_{N-1}}\left[
\Delta_{N}(f_{N-1})\right]  ^{p_{N-1}}dm_{N-1}\right]  ^{i(p_{N-1})}%
;~f_{N-1}\in E_{N-1} \label{d3.8}%
\end{align}
By Proposition \ref{markerRemarkDistFunctionTriangleHomogeneity},
$\Delta_{N-1}$ is a pseudo-norm with $\sigma_{N-1}=p_{N-1}i(p_{N-1}%
)\ \sigma_{N}$, so $E_{N-1}$ is a MVS.

Now we proceed to the induction step. Consider that it is proven that
$\Delta_{k}$ is a pseudo-norm with homogeneity degree $\sigma_{k}$ and
$\ (E_{k},\Delta_{k})\ \ $\ is a MVS, with the structure:
\begin{align}
E_{k}  &  =L^{p_{k}}(\Omega_{k},m_{k},E_{k+1},\Delta_{k+1})\label{d3.9}\\
\Delta_{k}(f_{k})  &  =\left[  \int\limits_{\Omega_{k}}\left[  \Delta
_{k+1}(f_{k})\right]  ^{p_{k}}dm_{k}\right]  ^{i(p_{k})};~f_{k}\in E_{k}
\label{d3.10}%
\end{align}
according to the induction hypothesis. We construct again the MVS
$(E_{k-1},\Delta_{k-1})$ : we set in Definition \ref{markerNotationDistanceLp}
\ $\Omega=\Omega_{k-1}$, $m=m_{k-1}$, $p=p_{k-1}$, and we select, $F=E_{k}$,
\ $\delta(.)=\Delta_{k}(.)$ from Eqs.(\ref{d3.9}, \ref{d3.10}). We obtain
\begin{align}
E_{k-1}  &  =L^{p_{k-1}}(\Omega_{k-1},m_{k-1},E_{k},\Delta_{k})\label{d.11}\\
\Delta_{k-1}(f_{k-1})  &  =\left[  \int\limits_{\Omega_{k-1}}\left[
\Delta_{k}(f_{k-1})\right]  ^{p_{k-1}}dm_{k-1}\right]  ^{i(p_{k-1})}%
;~f_{k-1}\in E_{k-1} \label{d3.12}%
\end{align}
By Proposition \ref{markerRemarkDistFunctionTriangleHomogeneity} $\Delta
_{k-1}$ is a pseudo-norm with $\sigma_{k-1}=p_{k-1}i(p_{k-1})\ \sigma_{k}$, so
$E_{k-1}$ is a MVS, that completes the induction step. The final MVS is
$E_{1}=L^{p_{1}}(\Omega_{1},m_{1},E_{2},\Delta_{2})$

By continuing this procedure down to $k=1$ \ we obtain the pseudo-norm
$\Delta_{1}(f_{1})\equiv D_{\mathbf{p},m}\left[  f_{N}\right]  $ with
$\sigma_{1}=\prod\limits_{k=1}^{N}p_{k}i(p_{k})$ that completes the proof.

\bigskip

\subsection{ Logarithmic convexity related to the R\'{e}nyi and generalized
R\'{e}nyi entropies \label{markerAppndxLogConvexity}}

Typical examples of log-convex functions are the integrals that define the
$L^{p}$ norms as functions of the exponent $p$. Consequently, many of the
properties of the Tsallis and R\'{e}nyi, as well as the generalized R\'{e}nyi
entropies can be derived from the properties of log-convex functions. We
mention that in the case of single variable all of the results can be reduced
to the well known facts \cite{EStein}. Nevertheless we give here an
self-contained treatment including also the case of single variable, because
only our approach can be extended to the case of many variables. We have the
following general

\begin{definition}
A non-negative real valued function $g(\mathbf{w})=g(w_{1},...,w_{n})$,
defined in the vector space $\mathbb{R}^{n}$ is log-convex if all of the one
variable functions $w_{k}\rightarrow\log g(a_{1},..a_{k-1},w_{k}%
,a_{k+1},...,a_{n})$ are convex functions. Equivalently, every log-convex
function can be represented as $\exp k(w_{1},...,w_{n})$ where $k(w_{1}%
,...,w_{n})$ is a convex function \textbf{separately} in each of the
variables, not necessary in the ensemble of all of the variables.
\end{definition}

We emphasize that the definition for many variables is adapted to our problem
and it might differ from definitions in the standard textbooks.

\subsubsection{Single variable}

This part is used for the proof of results related to R\'{e}nyi and Tsallis
entropies, and it is the starting point for higher dimensional generalizations.

For a function of a single variable $g(w)$ that has continuos second
derivative a necessary and sufficient condition for log-convexity is
\begin{equation}
\frac{d^{2}}{dw^{2}}\log g(w)=\frac{g(w)g^{\prime\prime}(w)-g^{\prime}(w)^{2}%
}{g^{2}(w)}\geq0 \label{logconv1}%
\end{equation}
The most important fact is the following well known interpolation property
\cite{EStein}. For easy reference, we give an elementary treatment

\begin{proposition}
\label{markPropLogConv1varBound} Let $w\rightarrow g(w)$ be a non-negative
function of real variable $w$, defined at least on the interval $[a,b]$. If
$g(w)$ is log-convex and
\begin{equation}
g(a)\leq A;~g(b)\leq B~ \label{logconv2}%
\end{equation}
then for all $w\in\lbrack a,b]$ we have the bound
\begin{equation}
g(w)\leq A^{\frac{b-w}{b-a}}B^{\frac{w-a}{b-a}} \label{logconv3}%
\end{equation}

\end{proposition}

\begin{proof}
The function $k(w)=\log f(w)$ is convex, so
\begin{equation}
k(\alpha a+(1-\alpha)b)\leq\alpha k(a)+(1-\alpha)k(b) \label{logconv4}%
\end{equation}
$\ \ $ where $0\leq\alpha\leq1$. Take $\alpha=(b-w)/(b-a)$ and from
Eq.(\ref{logconv4}) results
\[
\log g(w)=k(w)\leq\frac{b-w}{b-a}k(a)+\frac{w-a}{b-a}k(b)
\]
that with Eq.(\ref{logconv2}) results to be
\[
\log g(w)=k(w)\leq\frac{b-w}{b-a}\log A+\frac{w-a}{b-a}\log B
\]
that completes the proof.
\end{proof}

We have the following

\begin{proposition}
\label{markPropApendBasiclog_conv} Consider in the measure space
$(\Omega,\mathcal{A},m)$ the family of functions $f(\mathbf{x},w)$, such that
for every fixed value of the variable\textbf{ }$\mathbf{x}=\mathbf{x}_{0}%
\in\Omega$ the one variable function $w\rightarrow f(\mathbf{x}_{0},w)$ is
log-convex (consequently $f(\mathbf{x}_{0},w)>0$), for each fixed value of the
variable $w=w_{0}$, the function $\mathbf{x}\rightarrow f(\mathbf{x},w_{0})$
is integrable with respect to measure $dm(\mathbf{x})$ and the function
$g(w)$
\begin{equation}
g(w):=\int_{\Omega}dm(\mathbf{x})f(\mathbf{x},w) \label{logconv20}%
\end{equation}
has second derivative. Then $g(w)$ is log-convex. Equivalently, linear
combination\emph{ with non negative coefficients} of log-convex functions is
log-convex too
\end{proposition}

\begin{proof}
From the log-convexity of $f(x_{0},w)$ results that for some function
$k(x,w)$, which is convex in the variable $w$, we have
\begin{equation}
f(\mathbf{x},w)=\exp k(\mathbf{x},w) \label{logconv21}%
\end{equation}
According to Eq.(\ref{logconv1}) we have to prove
\begin{equation}
g(w)g^{\prime\prime}(w)-g^{\prime2}(w)\geq0 \label{logconv40}%
\end{equation}
From Eqs.(\ref{logconv20}, \ref{logconv21}, \ref{logconv40}) results
\begin{equation}
g(w)\int_{\Omega}dm(\mathbf{x})\left\{  \frac{\partial^{2}k}{\partial w^{2}%
}+\left[  \frac{\partial k}{\partial w}\right]  ^{2}\right\}  \exp
k(\mathbf{x},w)-\left[  \int_{\Omega}dm(\mathbf{x})\frac{\partial k}{\partial
w}\exp k(\mathbf{x},w)\right]  ^{2}\geq0 \label{logconv50}%
\end{equation}
The first term in Eq.(\ref{logconv50}) is always positive, since
$k(\mathbf{x},w)$ is convex in the variable $w$. For fixed $w$ we define a new
probability measure ($\int_{\Omega}dP(\mathbf{x})=1$) and introduce a new
notation $h(\mathbf{x})$
\begin{align}
dP(\mathbf{x})  &  :=dm(\mathbf{x})\frac{\exp k(\mathbf{x},w)}{g(w)}%
\label{logconv60}\\
h(\mathbf{x})  &  :=\frac{\partial k(\mathbf{x},w)}{\partial w}
\label{logconv61}%
\end{align}
By using the notations Eqs.(\ref{logconv60}, \ref{logconv61}) and
$\frac{\partial^{2}k}{\partial w^{2}}\geq0$, the proof is reduced to the well
known inequality between mean value and mean square value of the random
variable $h(\mathbf{x})$
\begin{equation}
\int_{\Omega}dP(\mathbf{x})\left[  h(\mathbf{x})\right]  ^{2}-\left[
\int_{\Omega}dP(\mathbf{x})h(\mathbf{x})\right]  ^{2}\geq0 \label{logconv70}%
\end{equation}
~that that proves Eq.(\ref{logconv50}).
\end{proof}

\begin{corollary}
\label{markCorllaryApendLogConvNorm} Consider a measure space $(\Omega
,\mathcal{A},m)$ and the real valued function $\phi(y)$ defined on $\Omega$,
such that the function
\begin{equation}
g(w):=\int_{\Omega}\left\vert \phi(x)\right\vert ^{w}dm(x) \label{logconv80}%
\end{equation}
is finite on some interval $a\leq w\leq b$. Then $g(w)$ is log-convex.
\end{corollary}

\begin{proof}
It follows from the previous Proposition \ref{markPropApendBasiclog_conv} with
$f(x,w)=\left\vert \phi(x)\right\vert ^{w}$ that is clearly log-convex.
\end{proof}

From Corollary \ref{markCorllaryApendLogConvNorm} and Proposition
\ref{markPropLogConv1varBound} it follows the following result. It can be
deduced from known properties of the $L^{p}$ norm in the case $1\leq p<q$, or
from the Hadamard three line theorem \cite{EStein}.

\begin{theorem}
\label{markerTheoremRieszThorinGen}Suppose that $\phi(x)\in L^{p}
(\Omega,m)\cap L^{q}(\Omega,m)$ with $0<p<q$ and
\begin{align}
\int_{\Omega}\left\vert \phi(x)\right\vert ^{p}dm(x)  &  \leq A_{p}%
\label{logconv80.1}\\
\int_{\Omega}\left\vert \phi(x)\right\vert ^{q}dm(x)  &  \leq A_{q}
\label{logconv80.2}%
\end{align}
Then $\phi(x)\in L^{r}(\Omega,m)$ with $0<p\leq r\leq q$ and
\begin{equation}
\int_{\Omega}\left\vert \phi(x)\right\vert ^{r}dm(x)\leq A_{q}^{\frac
{r-p}{q-p}}A_{p}^{\frac{q-r}{q-p}} \label{logconv80.3}%
\end{equation}

\end{theorem}

\begin{corollary}
\label{markeerCorolaryApendAnalitictyinStrip} Under the condition of the
previous theorem \ref{markerTheoremRieszThorinGen} there exists an unique
analytic function $F(z)=\int_{\Omega}\left\vert \phi(x)\right\vert ^{z}dm(x)$
defined in the strip \thinspace$p<\operatorname{Re}(z)<q$ such that
\begin{equation}
\left\vert F(x+iy)\right\vert \leq A_{q}^{\frac{x-p}{q-p}}A_{p}^{\frac
{q-x}{q-p}};~p\leq x\leq q \label{logconv80.4}%
\end{equation}

\end{corollary}

The previous results are sufficient to have a simple proof of the
interpolation Lemma\ref{markerInterpolationLemma}

\begin{proof}
\label{markerProofInterpLemma} In the previous Theorem
\ref{markerTheoremRieszThorinGen} we put $\phi(x)=f(x)$, $p=\min(1,s)$ and
$q=\max(1,s)$
\end{proof}

\subsubsection{Proof of the stability of the BSE (Theorem
\ref{markerTheoremShannonStability})
\label{markerSubsectAppendShannonStabilityProof}}

In the first step of the proof, we shall obtain convergence bounds on the
function $\int_{\Omega}\left\vert \rho_{n}(x)\right\vert ^{r}dm(x)-\int
_{\Omega}\left\vert \rho(x)\right\vert ^{r}dm(x)$ on a subset of the interval
$[p,q]$. From the bounds Eqs.(\ref{sh1}-\ref{sh2.1}) and triangle inequality
Eq.(\ref{LL4.4}), by simple algebra, we obtain the bounds
\begin{align}
\int\limits_{\Omega}\left\vert \rho_{n}(x)-\rho(x)\right\vert ^{p}dm(x)  &
\leq B_{p}:=2A;~p<1\label{proofS1}\\
\int\limits_{\Omega}\left\vert \rho_{n}(x)-\rho(x)\right\vert ^{q}dm(x)  &
\leq B_{q}:=2^{q}A;~q>1 \label{proofS2}%
\end{align}
where $B_{p},B_{q}$ are constants. By using Eqs.~(\ref{sh0}, \ref{proofS1})
and Theorem \ref{markerTheoremRieszThorinGen}, with $\Phi(x)=\rho_{n}%
(x)-\rho(x)$ we obtain, for all $r$ with $p\leq r\leq1$
\begin{equation}
\int_{\Omega}\left\vert \rho_{n}(x)-\rho(x)\right\vert ^{r}dm(x)\leq
(\varepsilon_{n})^{\frac{r-p}{1-p}}B_{p}^{\frac{1-r}{1-p}} \label{proofS3}%
\end{equation}
Similarly, from Eqs.(\ref{sh0}, \ref{proofS2}) and Theorem
\ref{markerTheoremRieszThorinGen}, for all $r$ in the domain $1\leq r\leq q$
the following inequality is obtained
\begin{equation}
\int_{\Omega}\left\vert \rho_{n}(x)-\rho(x)\right\vert ^{r}dm(x)\leq
B_{q}^{\frac{r-1}{q-1}}(\varepsilon_{n})^{\frac{q-r}{q-1}} \label{proofS4}%
\end{equation}
In the following we will restrict ourselves to a smaller domain: let
$p_{1}=(1+p)/2$ and $\ q_{1}=(1+q)/2$. Denote
\begin{align}
d_{1,n}  &  :=\underset{r\in\lbrack p_{1},1]}{\max}(\varepsilon_{n}%
)^{\frac{r-p}{1-p}}B_{p}^{\frac{1-r}{1-p}}\label{proofS4.1}\\
d_{2,n}  &  :=\underset{r\in\lbrack1,q_{1}]}{\max}B_{q}^{\frac{r-1}{q-1}%
}(\varepsilon_{n})^{\frac{q-r}{q-1}} \label{proofS4.2}%
\end{align}
and remark that from Eqs.(\ref{proofS4}, \ref{sh0}), for all $r$, in the
domain $p_{1}\leq r\leq1$, the following uniform convergence is found
($p_{1}<1<q_{1}$)
\begin{equation}
\int_{\Omega}\left\vert \rho_{n}(x)-\rho(x)\right\vert ^{r}dm(x)\leq
d_{1,n}\underset{n\rightarrow\infty}{\rightarrow}0 \label{proofS6}%
\end{equation}
Denote
\begin{equation}
G_{n}(r):=\int_{\Omega}\left\vert \rho_{n}(x)\right\vert ^{r}dm(x)-\int
_{\Omega}\left\vert \rho(x)\right\vert ^{r}dm(x) \label{Gn}%
\end{equation}
that, according to the Corollary \ref{markeerCorolaryApendAnalitictyinStrip},
can be analytically continued to the strip $\{z|p\leq\operatorname{Re}z\leq
q\}$. In the domain $p_{1}\leq r\leq1$, from the inequality (\ref{LL4.41}) and
Eq.(\ref{proofS6}) results for all $p_{1}\leq r\leq1$
\begin{equation}
\left\vert G_{n}(r)\right\vert \leq\int_{\Omega}\left\vert \rho_{n}%
(x)-\rho(x)\right\vert ^{r}dm(x)\leq d_{1,n}\underset{n\rightarrow\infty
}{\rightarrow}0 \label{proofS7}%
\end{equation}
In a similar manner, in the domain $1\leq r\leq q_{1}$, from the
inequality(\ref{LL4.41}) and Eq.(\ref{proofS6}) we have $\left\Vert \rho
_{n}-\rho\right\Vert _{r}\leq\left[  d_{2},_{n}\right]  ^{1/r}$ and from
Eq.(\ref{LL4.41}) we get
\begin{align}
\left\vert \left\Vert \rho_{n}\ \right\Vert _{r}-\left\Vert \rho\right\Vert
_{r}\right\vert  &  \leq\left[  d_{2,n}\right]  ^{1/r}\underset{n\rightarrow
\infty}{\rightarrow}0\label{proofS8}\\
\left\Vert \rho_{n}\ \right\Vert _{r}  &  =\left[  \int_{\Omega}\left\vert
\rho_{n}(x)\right\vert ^{r}dm(x)\right]  ^{1/r};~r\geq1 \label{proofS9}%
\end{align}
We have to obtain bounds on $G_{n}(r)$ on the interval $1\leq r\leq q_{1}$.
From Eq.(\ref{proofS8}) results that there exists $N$ sufficiently large such
that for all $n>N$ we have
\begin{equation}
A_{r}\leq\left\Vert \rho_{n}\ \right\Vert _{r}\leq B_{r} \label{proofS21}%
\end{equation}
where we denoted
\begin{align}
A_{r}  &  =\left\Vert \rho\right\Vert _{r}-d_{2,N}^{1/r}>0\label{proofS22}\\
B_{r}  &  =\left\Vert \rho\right\Vert _{r}+d_{2,N}^{1/r} \label{proofS223}%
\end{align}
Note that if $r\geq1$, $A_{r}\leq x\leq B_{r}$,~ $A_{r}\leq y\leq B_{r}$ \ we
have the following algebraic inequality
\begin{equation}
\left\vert x^{r}-y^{r}\right\vert \leq rB_{r}^{r-1}\left\vert x-y\right\vert
\label{proofS24}%
\end{equation}
By setting $x=\left\Vert \rho\right\Vert _{r}$ and $y=\left\Vert \rho
_{n}\right\Vert _{r}$ and by using Eqs.(\ref{proofS8}, \ref{proofS24}) and the
notation Eq.(\ref{Gn}), we obtain, for all $1\leq r\leq q_{1}$
\begin{equation}
\left\vert G_{n}(r)\right\vert \leq rB_{r}^{r-1}\left[  d_{2,n}\right]  ^{1/r}
\label{proofS5}%
\end{equation}
(Hint: use the mean value theorem for the function $x\rightarrow x^{r}$).
Combined with Eq.(\ref{proofS7}), we obtain the following bound in the domain
$p_{1}\leq r\leq q_{1}$
\begin{align}
\left\vert G_{n}(r)\right\vert  &  \leq\delta_{n}\label{proofS25}\\
\delta_{n}  &  =\max(d_{3,n},d_{1,n})\label{proofS26}\\
d_{3,n}  &  =\underset{1\leq r\leq q_{1}}{\max}rB_{r}^{r-1}\left[
d_{2,n}\right]  ^{1/r} \label{profS27}%
\end{align}
From Eq.(\ref{proofS7}, \ref{proofS8}, \ref{proofS5}-\ref{profS27}) the
following uniform convergence bound results
\begin{align}
\left\vert G_{n}(r)\right\vert  &  \leq\delta_{n}\underset{n\rightarrow\infty
}{\rightarrow}0\label{proofS29}\\
r  &  \in\Gamma_{0}:=[p_{1},q_{1}] \label{proofS30}%
\end{align}
Let $b:=q-p>0$ and denote
\begin{align*}
\Gamma_{0}  &  :=\{z|p_{1}\leq\operatorname{Re}z\leq q_{1},\operatorname{Im}%
z=0\}\\
\Gamma_{1}  &  :=\{z|z=p_{1}+it;\ 0\leq t\leq b\}\\
\Gamma_{2}  &  :=\{z|z=q_{1}+it;\ 0\leq t\leq b\}\\
\Gamma_{3}  &  :=\{z|p_{1}\leq\operatorname{Re}z\leq q_{1},\operatorname{Im}%
z=b\}\\
\Gamma_{m}  &  :=\Gamma_{1}\cup\Gamma_{2}\cup\Gamma_{3}%
\end{align*}
and define the domain $D:=\{z|p_{1}<\operatorname{Re}z<q_{1}%
,~\ 0<\operatorname{Im}z<b\}$ that is contained in the domain of holomorphy of
$G_{n}$. Now we prepare to use the previous extrapolation Theorem
\ref{markerTheoremExtrapol}. According to the Corollary
\ref{markeerCorolaryApendAnalitictyinStrip} $G_{n}(r)$ can be analitically
continued to the strip $\{z|p\leq\operatorname{Re}z\leq q\}$. From
Eqs.(\ref{sh2.2}, \ref{sh2.3}, \ref{LL4.3}) we get
\begin{equation}
\left\vert G_{n}(z)\right\vert \leq m=2A;\ z\in D\cup\Gamma_{0}\cup\Gamma_{m}
\label{proofS311}%
\end{equation}
Similarly to the proof of the Proposition
\ref{markerPropositionAnaliticityBound}, consider now a circle $C~$with the
center at $z=1$ and having radius $R=\min((1-p_{1})/2,~(q_{1}-1)/2)$. From
Eqs.(\ref{Gn}, \ref{sh4}) we obtain
\begin{align*}
S_{cl}[\rho]  &  =-\left[  \frac{d}{dz}F(z)\right]  _{z=1}=-\frac{1}{2\pi
i}\oint\limits_{C}\frac{F(w)dw}{(w-1)^{2}}\\
F(z)  &  :=\int_{\Omega}\left\vert \rho(x)\right\vert ^{z}dm(x)
\end{align*}
and finally results
\begin{equation}
S_{cl}[\rho_{n}]-S_{cl}[\rho]=-\frac{1}{2\pi i}\oint\limits_{C}\frac
{G_{n}(w)dw}{(w-1)^{2}} \label{proofS32}%
\end{equation}
It is clear that the convergence of the classical entropy is controlled by
\begin{equation}
\left\vert S_{cl}[\rho_{n}]-S_{cl}[\rho]\right\vert \leq\frac{1}{R}%
\underset{w\in C}{\max}\left\vert G_{n}(w)\right\vert \label{proofS33}%
\end{equation}
Now we use the Theorem \ref{markerTheoremExtrapol}, with the input data given
in Eqs.(\ref{proofS29}, \ref{proofS30}, \ref{proofS311}). Then for all $w$ in
the half circle in the upper complex half plane we have
\begin{equation}
\left\vert G_{n}(w)\right\vert \leq K\exp[\log(\delta_{n})~u_{0}%
(w)];\ |w-1|=R;~\operatorname{Im}(w)\geq0 \label{proofS34}%
\end{equation}
where the harmonic function $u_{0}(w)$ was defined by Eqs.(\ref{ae6.1},
\ref{ae6.2}), and $K$ is a constant. Since $G_{n}(w)$ is real in the real
axis, the previous Eq.(\ref{proofS34}) extends also in the lower semicircle.
Since in the interior of the domain $D$ the harmonic measure $u_{0}(w)$ is
strictly positive, we have $G_{n}(w)\rightarrow0$. From compactness of the
circle, the uniform convergence $\underset{w\in C}{\max}\left\vert
G_{n}(w)\right\vert \rightarrow0$ is ensured, and finally we obtain that
$\left\vert S_{cl}[\rho_{n}]-S_{cl}[\rho]\right\vert \underset{n\rightarrow
\infty}{\rightarrow}0$, which completes the proof.

\subsubsection{Logarithmic convexity properties and bounds in the case of many
variables.\label{markerSubsectLogCcnvManyVriables}}

Our definition of the log-convexity, in the case of many variables, differs
from the usual definition: In the generalization to many variables we shall
restrict our study to the case when a function of $N$ variables is log-convex
in one of variables while the remaining $N-1$ are fixed. Typical example of
interest is the nested integral that appear in the definition of the GRE. We
will study the function of many variables $\mathbf{w}=(w_{1,} w_{2}%
,...,w_{N})$, that is defined recurrently in Eqs.(\ref{LL20}-\ref{LL25}). We
set
\begin{equation}
g(w_{1,}w_{2},...,w_{N}):=N_{\mathbf{w},m}[\rho] \label{logconv81}%
\end{equation}
By using the previous Corollary \ref{markCorllaryApendLogConvNorm} it is clear
that for fixed $w_{2}^{(0)},...,w_{N}^{(0)}$, the function $w\rightarrow
g(w,w_{2}^{(0)},...,w_{N}^{(0)})$\ is log-convex. We have the following more
general result

\begin{theorem}
\label{markerTheoremAppLogConvNvar}For all $1\leq k\leq N$, for fixed
$w_{1}^{(0)}=q_{1},...,w_{k-1}^{(0)}=q_{k-1},w_{k+1}^{(0)}=q_{k+1}%
,...,w_{N}^{(0)}\ =q_{N}$, the function $w\rightarrow g(w_{1}^{(0)}%
,...,w_{k-1}^{(0)},w,w_{k+1}^{(0)},...,w_{N}^{(0)}):=h(w)$ is log-convex in
the variable $w$.
\end{theorem}

\begin{proof}
The proof is by induction. We rewrite the part of interest from the reccurence
relations Eqs.(\ref{LL20}-\ref{LL25}) as follows
\begin{align}
\rho_{k}^{\prime}(x_{1},x_{2},...,x_{k})  &  :=\int\limits_{\Omega_{k}}\left[
\rho_{k+1\text{ }}^{\prime}(x_{1},...,x_{k+1})\right]  ^{q_{k+1}}%
dm_{k+1}(x_{k+1})\\
\rho_{k-1}^{\prime}(w,x_{1},x_{2},...,x_{k-1})  &  :=\int\limits_{\Omega_{k}%
}\left[  \rho_{k\text{ }}^{\prime}(x_{1},...,x_{k})\right]  ^{w}dm_{k}%
(x_{k})\label{logconv90}\\
\rho_{k-2}^{\prime}(w,x_{1},x_{2},...,x_{k-2})  &  :=\int\limits_{\Omega_{k}%
}\left[  \rho_{k\text{ }}^{\prime}(x_{1},...,x_{k})\right]  ^{q_{k-1}}%
dm_{k}(x_{k})\label{logconv100}\\
&  ...\\
\rho_{1\text{ }}^{\prime}(w,\mathbf{x}_{1})  &  :=\int\limits_{\Omega_{2}%
}\left[  \rho_{2\text{ }}^{\prime}(w,x_{1},x_{2})\right]  ^{q_{2}}dm_{2}%
(x_{2})\label{logconv110}\\
h(w)  &  :=\int\limits_{\Omega_{1}}\left[  \rho_{1\text{ }}^{\prime}%
(w,x_{1})\right]  ^{q_{1}}dm_{1}(x_{1}) \label{logconv120}%
\end{align}
From Corollary \ref{markCorllaryApendLogConvNorm} results that the function
$\rho_{k-1}^{\prime}(w,x_{1},x_{2},...,x_{k-1})$, given by Eq.(\ref{logconv90}%
), is log-convex in the variable $w$. By successive application of the
Proposition \ref{markPropApendBasiclog_conv} to Eqs.(\ref{logconv100}%
-\ref{logconv120}) and by observing that any positive power of a log convex
function is log-convex, we conclude by induction, successively, that the
functions $\rho_{k-2}^{\prime}(w,x_{1},x_{2},...,x_{k-2})\,$,...,
$\rho_{1\text{ }}^{\prime}(w,\mathbf{x}_{1})$, $h(w)$ are all log-convex in
the variable $w$.
\end{proof}

For the sake of clarity we consider first the case $N=2$. Suppose now that
$g(w_{1},w_{2})$ is log-convex in the variables $w_{1}$, $w_{2}$, in the
rectangular domain $D_{2}$
\begin{equation}
a_{1}^{(1)}\leq w_{1}\leq a_{1}^{(2)};~~a_{2}^{(1)}\leq w_{2}\leq a_{2}^{(2)}
\label{logconv130}%
\end{equation}
and we have the following bounds in the corner points of $D_{2}$
\begin{align}
g\left(  a_{1}^{(1)},a_{2}^{(1)}\right)   &  \leq A_{1,1}\label{logconv131}\\
g\left(  a_{1}^{(2)},a_{2}^{(1)}\right)   &  \leq A_{2,1}\label{logconv132}\\
g\left(  a_{1}^{(1)},a_{2}^{(2)}\right)   &  \leq A_{1,2}\label{logconv133}\\
g\left(  a_{1}^{(2)},a_{2}^{(2)}\right)   &  \leq A_{2,2} \label{logconv134}%
\end{align}
We denote
\begin{align}
\rho_{1}^{(1)}(w_{1})  &  =\frac{a_{1}^{(2)}-w_{1}}{a_{1}^{(2)}-a_{1}^{(1)}
}\label{logconv135}\\
\rho_{1}^{(2)}(w_{1})  &  =\frac{w_{1}-a_{1}^{(1)}}{a_{1}^{(2)}-a_{1}^{(1)}
}\label{logconv136}\\
\rho_{2}^{(1)}(w_{2})  &  =\frac{a_{2}^{(2)}-w_{2}}{a_{2}^{(2)}-a_{2}^{(1)}
}\label{logconv137}\\
\rho_{2}^{(2)}(w_{2})  &  =\frac{w_{2}-a_{2}^{(1)}}{a_{2}^{(2)}-a_{2}^{(1)}}
\label{logconv138}%
\end{align}

\begin{remark}
\label{markerRemarkRhoPozitivity}If $\mathbf{w}=\{w_{1},w_{2}\}$ is an
interior point of the rectangle $D_{2}$ then $0<\rho^{(i)}(w_{k})<1$, for all
$1\leq i\leq2$ and $1\leq k\leq2$
\end{remark}

By using successively the Proposition \ref{markPropLogConv1varBound}, we
define the following bound in the interior point with coordinates
$(w_{1},w_{2})$, in terms of values$_{{}}$ on the corner points of the
rectangle Eq.(\ref{logconv130}):
\begin{align}
\log~g(w_{1},w_{2})  &  \leq b_{2}(w_{1},w_{2}):=b_{2}(\mathbf{w}
)\label{logconv140}\\
b_{2}(\mathbf{w})  &  =P_{1,1}(\mathbf{w})\log A_{1,1}+\ P_{1,2}
(\mathbf{w})\log A_{1,2}\label{logconv141}\\
&  +P_{2,1}(\mathbf{w})\log A_{2,1}+P_{2,1}(\mathbf{w})\log A_{2,2}\\
P_{m_{1},m_{2}}(\mathbf{w})  &  =P_{m_{1},m_{2}}(w_{1},w_{2})\mathbf{=}
\rho_{1}^{(m_{1})}(w_{1})\rho_{2}^{(m_{2})}(w_{2});\ \ m_{1},m_{2}
=\overline{1,2} \label{logconv142}%
\end{align}
From the previous remark results $0<P_{m_{1},m_{2}}(\mathbf{w})<1$ if
$\mathbf{w}$ is an interior point of the rectangle $D_{2}$

From Eq.(\ref{logconv140}) we obtain the following generalization of
Proposition \ref{markPropLogConv1varBound}

\begin{proposition}
\label{markPropApendLogConvBoundedConvergence} Suppose that we have the
sequence of n functions $g_{n}(w_{1},w_{2})$ that are log convex (hence non
negative) in the domain Eq.(\ref{logconv130}) and on the $4$ corner points we
have the uniform bound
\begin{equation}
g_{n}(a_{1}^{(i)},a_{2}^{(j)})\leq A_{i,j};~i,j=\overline{1,2}
\label{logconv150}%
\end{equation}
Then in the all interior points of the rectangle Eq.(\ref{logconv130}) we have
the uniform bound
\[
g_{n}(w_{1},w_{2})\leq\exp b(w_{1},w_{2})
\]
with $b(w_{1},w_{2})$ given by Eq.(\ref{logconv141}). If in at least in one of
the corner points $(a_{1}^{(i_{0})},a_{2}^{(j_{0})})$ we have
\begin{equation}
g_{n}(a_{1}^{(i_{0})},a_{2}^{(j_{0})})\rightarrow0 \label{logconv160}%
\end{equation}
and in the rest of the corner points we have the uniform bound
Eq.(\ref{logconv150})\ then for all interior points $(w_{1},w_{2})$ we have
also
\begin{equation}
g_{n}(w_{1},w_{2})\rightarrow0 \label{logconv170}%
\end{equation}
for all \emph{interior} points in the rectangle defined by
Eq.(\ref{logconv130}).
\end{proposition}

Now we extend the previous result for arbitrary number of variables. The
relations that follows in the particular case $N=2$ reduces to
Eqs.(\ref{logconv130}-\ref{logconv170}). Suppose that $g(w_{1},w_{2}%
,...,w_{N})$ is log-convex in the variables $w_{1},w_{2},...,w_{N}$, in the
$N$ dimensional hyper rectangle domain $D_{N}\subset\mathbb{R}^{N}$
\begin{equation}
a_{k}^{(1)}\leq w_{k}\leq a_{k}^{(2)}~~~;~k=\overline{1,N} \label{logconv180}%
\end{equation}
Suppose that in the $2^{N}$ corner points of $D_{N}$ (the set of vertices
$V_{N}$ of the hyper-rectangle ) we have the bounds
\begin{equation}
g(a_{1}^{(v_{1})},a_{2}^{(v_{2})},...,a_{N}^{(v_{N})})\leq A_{v_{1}%
,v_{2},...,v_{N}};~~v_{j}=\overline{1,2};~~j=\overline{1,N} \label{logconv190}%
\end{equation}
or denoting $v:=(v_{1},...,v_{N})\in\{1,2\}^{N}$, the set of vertices $V_{N}$
are represented as follows:\ $(a_{1}^{(v_{1})},a_{2}^{(v_{2})},...,a_{N}%
^{(v_{N})}):=\mathbf{a}^{(\mathbf{v})}\in V_{N}$. In analogy with
Eqs.(\ref{logconv135}-\ref{logconv138}) we introduce the following notations
\begin{align}
\rho_{k}^{(1)}(w_{k})  &  =\frac{a_{k}^{(2)}-w_{k}}{a_{k}^{(2)}-a_{k}^{(1)}%
}~~;~k=\overline{1,N}\label{logconv200}\\
\rho_{k}^{(2)}(w_{k})  &  =\frac{w_{k}-a_{k}^{(1)}}{a_{k}^{(2)}-a_{k}^{(1)}%
}~~;~k=\overline{1,N} \label{logconv210}%
\end{align}
To obtain a more compact notation for the generalization of
Eqs.(\ref{logconv140}-\ref{logconv142}), we use the following new notations
for the set of constants $A_{v_{1},v_{2},...,v_{N}}$
\begin{equation}
A_{\mathbf{v}}^{\prime}:=A_{v_{1},v_{2},...,v_{N}}~~v_{j}=\overline
{1,2};~~j=\overline{1,N} \label{1000}%
\end{equation}
\ \ and define the function $P:V_{N}\times D_{N}\rightarrow\mathbb{R}$ as
follows \
\begin{align}
P(\mathbf{v},\mathbf{w})  &  :=\prod\limits_{k=1}^{N}\rho_{k}^{(v_{k})}%
(w_{k});~~v_{j}=\overline{1,2};~~j=\overline{1,N}\label{1100}\\
v  &  :=(v_{1},...,v_{N})\in\{1,2\}^{N}\label{1200}\\
\mathbf{w}  &  \mathbf{=}(w_{1},...,w_{N}\mathbf{)} \label{1300}%
\end{align}
By Remark \ref{markerRemarkRhoPozitivity} we have
\begin{equation}
0<P(\mathbf{v},\mathbf{w})<1;~\mathbf{w}\in Int(D_{N});~\mathbf{v\in V}_{N}
\label{1350}%
\end{equation}
where the set of interior points $Int(D_{N})\subset D_{N}$ is defined by
$a_{k}^{(1)}<w_{k}<a_{k}^{(2)}~~~;~k=\overline{1,N}$. By using successively
the Proposition \ref{markPropLogConv1varBound}, with the notations
Eqs.(\ref{1000}-\ref{1300}) \ the following bounds results

\begin{corollary}
\label{markerCorolaryLogconvBoundNvar}If the function $\mathbf{w}\rightarrow
g(\mathbf{w})$ is log-convex, then\ under previous conditions
Eqs.(\ref{logconv190}-\ref{logconv210}) we have the bound in the
hyper-rectangle $\ D_{N}$
\begin{equation}
\log~g(w_{1},w_{2},...,w_{N})\leq b_{N}(\mathbf{w}) \label{logconv220.1}%
\end{equation}
where we denote
\begin{equation}
b_{N}(\mathbf{w}):=\sum\limits_{\mathbf{v}\in V_{N}}P(\mathbf{v}%
,\mathbf{w})\log A_{\mathbf{v}}^{\prime}%
\end{equation}

\end{corollary}

\bigskip

\begin{remark}
It can be proven that $\sum\limits_{\mathbf{v}\in V_{N}}P(\mathbf{v}
,\mathbf{w})=1$
\end{remark}


\begin{thebibliography}{99}                                                                                               %


\bibitem {shannon}C. E. Shannon, \textit{A Mathematical Theory of
Communication}, \textit{Bell System Technical Journal}, \textbf{27}, 379-423
\& 623-656 (1948).

\bibitem {Renyi1}A. R\'{e}nyi (1960), \textit{On measures of information and
entropy}, \textit{Proceedings of the fourth Berkeley Symposium on Mathematics,
Statistics and Probability, June 20-July 30, 1960, Volume I, University of
California Press, Berkeley and Los Angeles, 547-561 1960}. pp. 547--561.

\bibitem {RenyiWT}A. R\'{e}nyi, \ Wahrscheinlichkeitstheorie (De Gruyter,
Berlin, 1974).

\bibitem {RenyiInfoAccumulation}A. R\'{e}nyi, \textit{On statistical \ laws of
accumulation of information}, in\textit{ R\'{e}nyi Alfr\'{e}d V\'{a}logatott
Munk\'{a}i (A. R\'{e}nyi Selected works)}, volume 3, 1962-1970,
Akad\textit{\'{e}}miai Kiad\'{o}, Budapest 1977, page 33.

\bibitem {Renyi4}A. R\'{e}nyi, Rev. Int. Inst. Stat., \textbf{33}, 1-14 (1965).

\bibitem {Tsallis1}C. Tsallis, \textit{Journal of Statistical Physics},
\textbf{52}, 479--487 (1988).

\bibitem {Tsallis2}C. Tsallis, R. S. Mendes and A. R. Plastino, Physica
\textbf{A, 261}, 534-554 (1998).

\bibitem {TsallisBook}C. Tsallis, \textit{Introduction to nonextensive
statistical mechanics\_ approaching a complex world, (}Springer, 2009).

\bibitem {TsallisGelMann}M. Gell-Mann, C. Tsallis, \textit{Nonextensive
entropy- Interdisciplinary applications,} (Oxford University Press, USA, 2004).

\bibitem {SonninoSteinbrGRE}G. Sonnino, G. Steinbrecher, Phys. Rev. \textbf{E}
\textbf{89}, 062106 (2014).

\bibitem {Klimontovich}Yu. L. Klimontovich, Chaos, Solitons, Fractals,
\textbf{5}, 1985-2002 (1995).

\bibitem {Csiszar}I. Csisz\'{a}r, I.E.E.E. Transactions on Information Theory,
\textbf{41} -1, 26, (1995).

\bibitem {SGYSGstructure}G. Steinbrecher, G. Sonnino, "\textit{Generalized
R\'{e}nyi Entropy and Structure Detection of Complex Dynamical Systems}",
arXiv:1512.06108v1 [physics.data-an] (2015).

\bibitem {maszczyk}T. Maszczyk and X. Duch, \textit{Comparison of Shannon,
R{\'{e}}nyi and Tsallis Entropy used in Decision Trees}, \textit{Lecture Notes
in Computer Science}, \textbf{5097}, 643 (2008).

\bibitem {RenyiDivergence}T. van Erven, P. Harremo\"{e}s,\textit{ "}R\'{e}nyi
Divergence and Kullback-Leibler Divergence"\textit{,} e-print
arXive:1206.2459v2\textit{ \ }(2014).

\bibitem {SGySASGCategory}G. Steinbrecher, A. Sonnino, G. Sonnino, Journal of
Modern Physics,\textbf{ 7}, 251-266 (2016).

\bibitem {CLTapproximation}Hongfei Cui, Jianqiang Sun, Yiming Ding,
\textit{The rates of convergence for generalized entropy of the normalized
sums of IID random variables}, arXiv:1106.3381v1 [cs.IT] (2011).

\bibitem {Lesche}B. Lesche, Journal of Statistical Physics, \textbf{\ 27}, 419-422,\ \ (1982).

\bibitem {Rudin}W. Rudin (1987), Real and Complex Analysis, McGraw Hill Inc.
3rd Ed. page 74.

\bibitem {ReedSimon}M. Reed and B. Simon, Functional Analysis (Methods of
Modern Mathematical Physics). Vol. 1, Academic Press (1981).

\bibitem {EStein}E. M. Stein, R. Shakarchi, \textit{Functional Analysis,
Introduction to further topics, Princeton Lectures in Analysis IV, Princeton
Univ. Press, Princeton (2011),page 36 exerc. 9.c. page 39 exerc. 20.}

\bibitem {CiulliNumerical1}M. Ciulli, S. Ciulli, Computer Physics
Communications, \textbf{18}, 215 (1981).

\bibitem {CiulliNumerical2}I. Caprini, M. S\u{a}raru, C. Pomponiu, M. Ciulli,
S. Ciulli, I. Sabba-\c{S}tef\~{a}nescu, Computer Physics Communications
\textbf{18}, 305-326 (1979).

\bibitem {CiulliStabProblems}S. Ciulli, \textit{Stability Problems in Analytic
Continuation}, \ Lectures given at the 1972 Int. Institute for theoretical
strong interaction physics, Kaiserslautern. Pag. 70-105, Springer Berlin
Heidelberg (1973).

\bibitem {CiulliNenciu}S. Ciulli, G. Nenciu, \ J. Math. Phys. \textbf{14},
1675 (1973).

\bibitem {CiulliPhysRep}S. Ciulli, C. Pomponiu, I. Sabba-\c{S}tef\~{a}nescu,
Physics Reports \textbf{17}, 133-224 (1975).

\bibitem {NenciuHamonicMeasure}G. Nenciu, Lettere al Nuovo Cimento,
\textbf{4}, Issue 3, pp. 96-100 (1970).

\bibitem {JizbaArimitsu}P. Jizba, T. Arimitsu, \ Phys. Rev. \textbf{E 69},
026128 (2004).

\bibitem {LuschgiPages}H. Luschgy and G. Pag\`{e}s, Electronic Communication
in Probability, \textbf{13}, 422-434 (2008).

\bibitem {Besov}O. V. Besov, V. P. Il'in and S.M. Nikol'skii, \textit{Integral
representations of functions and embedding theorems}, Ed. Nauka, pp. 9-40,
Moscow (in Russian) (1975).

\bibitem {ESteinSingularIntegrals}E. M. Stein, \textit{Singular Integrals and
Differentiable Properties of Functions. }Princeton Univ. Press. (1971). \ 

\bibitem {CLTmaxent}O. Johnson, \textit{Information Theory and the Central
Limit Theorem. }Imperial College Press, London (2004).

\bibitem {RLASATO}H. Takayasu, A.-H. Sato, and M. Takayasu, Phys. Rev. Lett.
\textbf{79}, 966 (1997).

\bibitem {sgw}G. Steinbrecher, X. Garbet and B. Weyssow, "\textit{Large time
behavior in random multiplicative processes}", arXiv:1007.0952v1, (2010).

\bibitem {sgbw}G. Steinbrecher, B. Weyssow, Phys. Rev. Lett. \textbf{92,
}12503 (2004); T. L. Rhodes et al., Phys. Lett. \textbf{A} \textbf{253}, 181
(1999).\textbf{\ }

\bibitem {intermittency}S. Auma\^{\i}tre, F. P\'{e}tr\'{e}lis and K. Mallick,
Phys. Rev. Lett. \textbf{95}, 064101 (2005).

\bibitem {BigData1}A. Chicocki, "\textit{Tensor Networks for Big Data
Analytics and Large-Scale Optimization Problems}", arXiv:1407.3124v2 [cs.NA] (2014).

\bibitem {Balescu1}M. Vlad, F. Spineanu, J. H. Misguich, and R. Balescu,
Phys.Rev. \textbf{E 58}, 7359 (1998).

\noindent E. Vanden Eijnden and R. Balescu, Phys. Plasmas \textbf{3}, 815 (1996).

\noindent M. Negrea, I. Petrisor, Phys.Rev. \textbf{E 70}, 046409 (2004).

\bibitem {balescu3}M. Vlad, F. Spineanu, J. H. Misguich, and R. Balescu,
Phys.Rev. \textbf{E 63}, 066304 (2001).

\bibitem {BakTangWiesenfeld}P. Bak, C. Tang and K. Wiesenfeld, Phys. Rev.
Lett. \textbf{59}, 381 (1987).

\bibitem {SOC}D. L. Turcotte, Rep. Prog. Phys. \textbf{62} 1377--1429 (1999).

\bibitem {DiamondHahm}P. H. Diamond and T. S. Hahm, Phys. Plasmas \textbf{2},
3640 (1995).

\bibitem {GarbetWaltz}X. Garbet and R. E. Waltz, Phys. Plasmas\textbf{ 5},
2836, (1998).

\bibitem {ConnorDandy}S. C. Chapman, R. O. Dendy, and B. Hnat, Phys. Rev.
Lett. \textbf{86},2814 (2001).

\bibitem {HcksCarreras}H. R. Hicks and B. A. Carreras, Phys. Plasmas\textbf{
8}, 3277, (2001).

\bibitem {GarciaCarreras}L. Garcia, B. A. Carreras, D. E. Newman, Phys.
Plasmas, \textbf{9}, 841, (2002).

\bibitem {YanickXavierGuillaume}G. Dif-Pradalier, P. H. Diamond, V.
Grandgirard, Y. Sarazin, J. Abiteboul, X. Garbet, Ph. Ghendrih, A. Strugarek,
S. Ku, C. S. Chang

Phys. Rev. \textbf{E 82}, 025401, (2010).

\bibitem {JizbaArimitsuAnalitic}P. Jizba, T. Arimitsu, Annals of Physics,
\textbf{312},17--59 (2004).

\bibitem {Nevanlinna}Rolf Nevanlinna, \textit{Analytic Functions},Springer
Berlin Heidelberg, pp. 26-45 (1970).

\bibitem {Oksendal}B. Oksendal, \textit{Stochastic Differential Equations},
Springer, Heidelberg, N.Y., p. 117 (2000).
\end{thebibliography}
\end{document}